\documentclass[acmsmall, review=false, nonacm=true]{acmart}
\setcopyright{none}
\renewcommand\footnotetextcopyrightpermission[1]{} 
\usepackage[normalem]{ulem}
\usepackage{enumitem}
\usepackage{listings}
\usepackage{arydshln}
\usepackage{balance}
\usepackage{bigstrut}
\usepackage{xspace}
\usepackage{mathtools}
\newcommand{\revision}[1]{#1}
\newcommand{\minirevision}[1]{#1}
\newcommand{\proofread}[1]{#1}
\theoremstyle{remark}
\newtheorem{claim}{Claim}
\newcommand{\naturals}{\mathbb{N}}
\newcommand{\rationals}{\mathbb{Q}}
\newcommand{\defeq}{\vcentcolon=}
\newcommand{\att}[1]{\mathit{#1}} 
\newcommand{\attu}[1]{\underline{\mathit{#1}}} 
\newcommand{\card}[1]{|#1|}
\newcommand{\vars}[1]{{\mathsf{vars}}(#1)}
\newcommand{\key}[1]{{\mathsf{Key}}(#1)}
\newcommand{\notkey}[1]{{\mathsf{notKey}}(#1)}
\newcommand{\calC}{\mathcal{C}}
\newcommand{\calF}{\mathcal{F}}
\newcommand{\AGG}{\mathtt{AGG}}
\newcommand{\SUM}{\mathtt{SUM}}
\newcommand{\PROD}{\mathtt{PRODUCT}}
\newcommand{\SUMDISTINCT}{\mathtt{SUM\mhyphen DISTINCT}}
\newcommand{\MAX}{\mathtt{MAX}}
\newcommand{\MIN}{\mathtt{MIN}}
\newcommand{\AVG}{\mathtt{AVG}}
\newcommand{\COUNT}{\mathtt{COUNT}}
\newcommand{\COUNTDISTINCT}{\mathtt{COUNT\mhyphen DISTINCT}}
\newcommand{\agagg}{\calF_{\AGG}} 
\newcommand{\agavg}{\calF_{\AVG}} 
\newcommand{\agmax}{\calF_{\MAX}} 
\newcommand{\agmin}{\calF_{\MIN}} 
\newcommand{\agsum}{\calF_{\SUM}} 
\newcommand{\agprod}{\calF_{\PROD}} 
\newcommand{\agsumdistinct}{\calF_{\SUMDISTINCT}} 
\newcommand{\agplus}{\calF_{\oplus} }
\newcommand{\agcount}{\calF_{\mathtt{COUNT}}} 
\newcommand{\agtermquantification}[2]{\mathsf{Aggr}_{#1}#2}
\newcommand{\agterm}[4]{\agtermquantification{#1}{#2}\left[{#3,#4}\right]}
\newcommand{\cqamodels}{\models_{\mathsf{cqa}}}
\newcommand{\db}{\mathbf{db}}
\newcommand{\dbstock}{\db_{\mathsf{Stock}}}
\newcommand{\repairs}[1]{\mathsf{rset}({#1})}
\newcommand{\rep}{\mathbf{r}}
\newcommand{\sep}{\mathbf{s}}
\newcommand{\tep}{\mathbf{t}}
\newcommand{\dom}{\mathbf{dom}}
\newcommand{\var}{\mathbf{var}}
\newcommand{\withaggr}[1]{\mathsf{AGGR}[#1]}

\newcommand{\sjfbcq}{\mathsf{sjfBCQ}}
\newcommand{\sort}{\prec}

\newcommand{\interpret}[2]{\llbracket{#1}\rrbracket^{#2}}
\newcommand{\fd}[2]{{#1}\rightarrow{#2}}
\newcommand{\substitute}[3]{{#1}_{[{#2}\mapsto{#3}]}}
\newcommand{\cqa}[1]{{\mathsf{CERTAINTY}}(#1)}
\mathchardef\mhyphen="2D
\newcommand{\glbcqa}[1]{{\mathsf{GLB\mhyphen CQA}}(#1)}
\newcommand{\lubcqa}[1]{{\mathsf{LUB\mhyphen CQA}}(#1)}

\newcommand{\attacks}[1]{\stackrel{#1}{\rightsquigarrow}}
\newcommand{\nattacks}[1]{\stackrel{#1}{\not\rightsquigarrow}}
\newcommand{\fdset}[1]{{\mathcal{K}}(#1)}
\newcommand{\keycl}[2]{{#1}^{+,#2}}
\newcommand{\problem}[1]{\textsf{#1}}
\newcommand{\frugalformula}[1]{\varphi_{#1}}
\newcommand{\rifi}[3]{{\mathsf{Reify}}({#1},{#2},{#3})}
\newcommand{\formula}[1]{(#1)}
\newcommand{\lrformula}[1]{\left(#1\right)}
\newcommand{\leftbag}{\left\{\!\left\{}
\newcommand{\rightbag}{\right\}\!\right\}}
\newcommand{\bag}[1]{\leftbag{#1}\rightbag}
\newcommand{\set}[1]{\left\{#1\right\}}
\newcommand{\bigset}[1]{\big\{{#1}\big\}}
\newcommand{\tuple}[1]{\left({#1}\right)}
\newcommand{\cclass}[1]{\mathsf{#1}}
\newcommand{\NP}{\cclass{NP}}
\newcommand{\coNP}{\cclass{coNP}}
\newcommand{\FO}{\cclass{FO}}
\newcommand{\NL}{\cclass{NL}}
\newcommand{\FP}{\cclass{FP}}
\newcommand{\PTIME}{\cclass{P}}
\newcommand{\LOGSPACE}{\cclass{L}}
\newcommand{\FOL}{\mathsf{FOL}}
\newcommand{\cforest}{\mathsf{Cforest}}
\newcommand{\caggforest}{\mathsf{Caggforest}}
\newcommand{\cnewclass}{\mathsf{Cparsimony}}
\newcommand{\domain}[1]{\mathsf{dom}(#1)}
\newcommand{\emptyvaluation}{\varnothing}
\newcommand{\restrict}[2]{{#1}\restriction_{#2}}
\newcommand{\dual}[1]{#1^{\mathsf{dual}}}
\newcommand{\powerset}[1]{2^{#1}}
\newcommand{\bluekill}{\hspace{-10em}\textcolor{blue}{\raisebox{0.2em}{\rule{10em}{0.1em}}}}
\newcommand{\redkill}{\hspace{-10em}\textcolor{red}{\raisebox{0.2em}{\rule{10em}{0.1em}}}}
\newcommand{\mymcs}[2]{\mathsf{MCS}_{#2}(#1)}
\newcommand{\copies}[2]{{#1}\#{#2}}
\newcommand{\myand}{\textnormal{\ and\ }}
\newcommand{\extendthree}[3]{\mathsf{Ext}\lrformula{#1\mid #2,#3}}
\newcommand{\extend}[1]{\mathsf{Ext}\lrformula{#1}}
\newcommand{\frugal}[2]{\preceq^{#1}_{#2}}
\newcommand{\myqed}{\hfill $\triangleleft$}

\title{Computing Range Consistent Answers to Aggregation Queries via Rewriting}
\author{Aziz Amezian El Khalfioui}
\email{aziz.amezianelkhalfioui@umons.ac.be}
\author{Jef Wijsen}
\email{jef.wijsen@umons.ac.be}
\orcid{https://orcid.org/0000-0001-8216-273X}
\affiliation{%
  \institution{University of Mons}
  \city{Mons}
  \country{Belgium}
}
\begin{document}

\begin{abstract}
    We consider the problem of answering conjunctive queries with aggregation on database instances that may violate primary key constraints.
    In SQL, these queries follow the SELECT-FROM-WHERE-GROUP BY format, where the WHERE-clause involves a conjunction of equalities, and the SELECT-clause can incorporate aggregate operators like MAX, MIN, SUM, AVG, or COUNT. 
    Repairs of a database instance are defined as inclusion-maximal subsets that satisfy all primary keys.
    For a given query, our primary objective is to identify repairs that yield the lowest aggregated value among all possible repairs. We particularly investigate queries for which this lowest aggregated value can be determined through a rewriting in first-order logic with aggregate operators.
\end{abstract}

\maketitle

\section{Introduction}\label{sec:introduction}

\emph{Consistent query answering} (CQA) was introduced at PODS'99~\cite{DBLP:conf/pods/ArenasBC99} as a principled approach to answering queries on database instances that are inconsistent with respect to a given set of integrity constraints.
The only integrity constraints we consider in the current work are primary keys. 
A \emph{block} in a database instance is a $\subseteq$-maximal set of tuples of a same relation~$R$ that agree on the primary key of~$R$.
A \emph{repair} of a database instance picks exactly one tuple from each block.
Given a Boolean query~$q$, $\cqa{q}$ is then defined as the decision problem that takes a database instance~$\db$ as input, and asks whether $q$ holds true in every repair of~$\db$. 
This problem, while commonly studied for Boolean queries, can be readily extended to queries with free variables~$\vec{x}$: a \emph{consistent answer} to a query $q(\vec{x})$ is any sequence~$\vec{c}$ of constants, of length $\card{\vec{x}}$, such that $q(\vec{c})$ holds true in every repair. 
The computational complexity of $\cqa{q}$ is well understood for all queries~$q$ in $\sjfbcq$, the class of self-join-free Boolean conjunctive queries~\cite{DBLP:journals/tods/KoutrisW17,DBLP:journals/mst/KoutrisW21}. This understanding readily extends to queries with free variables.


The current paper focuses on the complexity of CQA for \emph{numerical queries}~$g()$, which throughout this paper are mappings that take a database instance as input and return a single number; numerical queries are commonly called \emph{numerical terms} in logic. We will assume that some database columns are constrained to be numerical; in particular, values in numeric columns will be non-negative rational numbers throughout this paper, except for Section~\ref{sec:unconstrained}. We employ the range semantics as presented by Arenas et al.~\cite{DBLP:conf/icdt/ArenasBC01}, which provides the greatest lower bound (glb) and the least upper bound (lub) of query answers across all repairs.
Specifically, the function problems $\glbcqa{g()}$ and  $\lubcqa{g()}$ take a database instance~$\db$ as input, and return, respectively, the glb and the lub of the set that contains each number returned by~$g()$ on some repair.

As motivated by Arenas et al.~\cite{DBLP:conf/icdt/ArenasBC01}, range consistent query answers are of particular interest for aggregation queries, which are likely to return different numbers on different repairs, and therefore will lack a single consistent answer that holds true across all repairs.
Our attention will therefore be directed towards aggregation queries. In particular, we consider numerical queries $g()$ that take the following form in the extended Datalog syntax of~\cite{DBLP:conf/dmdw/CohenNS99, DBLP:journals/tods/CohenNS06}:
\begin{align}\label{eq:aggterm}
\AGG(r)\leftarrow q(\vec{u}),
\end{align}
where the \emph{body}~$q(\vec{u})$ is a conjunction of atoms (a.k.a.\ subgoals), $r$~is either a numeric variable occurring in~$\vec{u}$ or a constant rational number, and $\AGG$ is an aggregate symbol (like $\MAX$, $\MIN$, $\SUM$, $\AVG$, $\COUNT$). 
Such a query will be called called an $\AGG$-query, and is interpreted as follows.
Every aggregate symbol~$\AGG$ in our query language is associated with an aggregate operator, denoted~$\agagg$, which is a function that takes a multiset  of non-negative rational numbers, and returns a rational number. 
The semantics on a given database instance~$\db$ is standard:
let $\theta_{1}, \theta_{2},\ldots,\theta_{n}$ enumerate all embeddings of the body into~$\db$, then the query returns $\agagg(\bag{\theta_{1}(r), \theta_{2}(r), \ldots, \theta_{n}(r)})$.
Note that the argument of $\agagg$ is a multiset, because it is possible that $\theta_{i}(r)=\theta_{j}(r)$ for $i\neq j$.

Note that numerical queries~$g()$ of the form~\eqref{eq:aggterm} may return $\agagg(\emptyset)$.
The value of an aggregate operator on the empty multiset is often fixed by convention, such as setting $\agsum(\emptyset)\defeq 0$. 
To be independent of any such convention, we will define $\glbcqa{g()}$ and  $\lubcqa{g()}$ to return a distinguished constant~$\bot$ if $g()$ returns  $\agagg(\emptyset)$ on some repair. So we have the following function problem for any $\AGG$-query~$g()$ of the form~\eqref{eq:aggterm}:

\begin{center}
\framebox{
\begin{minipage}{0.9\columnwidth}
\begin{description}
     \item[Problem $\glbcqa{g()}$.]
    \item[Input:] A database instance $\db$ that may violate its primary key constraints.
    \item[Output:]
    \begin{tabular}[t]{l}
    Return $\bot$ if some repair of $\db$ falsifies $\exists\vec{u}\formula{q\formula{\vec{u}}}$; \\
    otherwise return $\min\set{\interpret{g()}{\rep}\mid\textnormal{$\rep$ is a repair of $\db$}}$.
    \end{tabular}
\end{description}
\end{minipage}}
\end{center}
Here, $\interpret{g()}{\rep}$ denotes the result of~$g()$ on~$\rep$.
The function problem $\lubcqa{q}$ is obtained by replacing $\min$ with $\max$ in the above problem.
Note that since every repair has a finite number of repairs, glb and lub are the same as min and max, respectively.


\begin{figure}\centering
\begin{small}
\begin{tabular}{cc}
\begin{tabular}{c|*{2}{c}}
$\att{Dealers}$ & $\attu{Name}$ & $\att{Town}$\bigstrut\\\cline{2-3}
$\dagger$ & Smith & Boston\\
          & Smith & New York\\\cdashline{2-3}
$\dagger$ & James & Boston
\end{tabular}
&
\begin{tabular}{c|*{3}{c}}
$\att{Stock}$ & $\attu{Product}$ & $\attu{Town}$ & $\att{Qty}$\bigstrut\\\cline{2-4}
$\dagger$ & Tesla X &  Boston & 35\\
          & Tesla X &  Boston & 40\\\cdashline{2-4}
$\dagger$ & Tesla Y &  Boston & 35\\\cdashline{2-4}
$\dagger$ & Tesla Y &  New York & 95\\
          & Tesla Y &  New York & 96        
\end{tabular}
\end{tabular}
\end{small}
\caption{Database instance $\dbstock$. Blocks are Separated by Dashed Lines.}
\label{fig:dealers}
\end{figure}

We write $\withaggr{\sjfbcq}$ for the class of all queries of the form~\eqref{eq:aggterm} whose body $q(\vec{u})$ is self-join-free (i.e., does not contain two distinct atoms with the same relation name).
In this paper, we are interested in the following complexity classification problem: Given a query~$g()$ in $\withaggr{\sjfbcq}$,
determine the computational complexity of the function problem $\glbcqa{g()}$.
For example, the database of Fig.~\ref{fig:dealers} records the quantity of products in stock in various towns (relation $\att{Stock}$) and the town of operation for each dealer (relation $\att{Dealers}$). The primary keys are underlined, and blocks are separated by dashed lines. The inconsistencies concern Smith's town of operation, and the stock levels of Tesla~X  and Tesla~Y in, respectively, Boston and New York.
For the example database of Fig.~\ref{fig:dealers}, the following query returns the total quantity of cars in stock in Smith's town of operation:
\begin{equation}\label{eq:numtermexample}\tag{$g_{0}()$}
\SUM(y)\leftarrow  \att{Dealers}(\underline{\textnormal{``Smith''}},t), \att{Stock}(\underline{p,t},y).
\end{equation}
In Fig.~\ref{fig:dealers}, the repair composed of the tuples preceded by $\dagger$ yields the answer $70$ ($=35+35$), which is the lowest answer achievable among all repairs.

From a theoretical perspective, our focus on the glb over the lub is without loss of generality, by exploiting that $\max(\bag{r_{1},r_{2},\ldots,r_{k}})=-1*\min(\bag{-1*r_{1},-1*r_{2},\ldots,-1*r_{k}})$.
For every aggregate operator~$\agagg$, we define its dual, denoted~$\dual{\agagg}$, by $\dual{\agagg}(X)\defeq-1*\agagg(X)$, for every multiset $X$. 
Then, solving $\lubcqa{g()}$ for an $\AGG$-query~$g()\defeq\AGG(r)\leftarrow q(\vec{u})$ is the same (up to a sign) as solving $\glbcqa{h()}$ for $h()\defeq\dual{\AGG}(r)\leftarrow q(\vec{u})$, where $\dual{\AGG}$ is the aggregate symbol associated with $\dual{\agagg}$.


\revision{We now discuss why we restrict ourselves to self-join-free queries.
Note that for a numerical query $g()\defeq\AGG(r)\leftarrow q(\vec{u})$, the decision problem $\glbcqa{g()}$ returns $\bot$ on (and only on)  ``no''-instances of $\cqa{\exists\vec{u}\formula{q(\vec{u})}}$. Thus,  a solution to $\glbcqa{g()}$ also solves $\cqa{\exists\vec{u}\formula{q(\vec{u})}}$. The computational complexity of the latter problem is well understood whenever $q(\vec{u})$ is self-join-free~\cite{DBLP:journals/tods/KoutrisW17}---an understanding we build upon in the current paper---but it remains largely unexplored for queries with self-joins. Only recently has it been established for self-joins of size~two~\cite{DBLP:journals/pacmmod/PadmanabhaSS24}. In particular, the concept of an attack graph, which is a key tool used in~\cite{DBLP:journals/tods/KoutrisW17} and in the current paper, loses its meaning and usefulness when self-joins are present.
Given the limited understanding of CQA for self-joins, we exclude self-joins in the current paper.} 

In our complexity study, we specifically aim to understand the conditions on~$g()$ under which $\glbcqa{g()}$ is solvable through rewriting in an aggregate logic, denoted $\withaggr{\FOL}$, which extends first-order logic with aggregate operators along the lines of~\cite{DBLP:journals/jacm/HellaLNW01}.
So our central problem takes as input a numerical query $g()$ in $\withaggr{\sjfbcq}$, and asks whether or not there is a numerical query~$\varphi()$ in $\withaggr{\FOL}$ that solves $\glbcqa{g()}$; moreover, when such $\varphi()$ exists, we are interested in constructing it,  a task loosely referred to as ``glb rewriting of $g()$ in $\withaggr{\FOL}$.''
A practical motivation for focusing on $\withaggr{\FOL}$ is that formulas in this logic are well-suited for implementation in SQL, allowing them to benefit from existing DBMS technology.

Our primary general result can now be stated as follows, with the definitions of \emph{monotone} and \emph{associative} deferred to Section~\ref{sec:aggregatingpos}:

\begin{theorem}[Separation Theorem]\label{the:glbmain}
The following decision problem is decidable \revision{in quadratic time (in the size of the input)}:
Given as input a numerical query~$g()$ in $\withaggr{\sjfbcq}$ whose aggregate operator is both monotone and associative,
is $\glbcqa{g()}$ expressible in $\withaggr{\FOL}$?
Moreover, if the answer is ``yes,'' then it is possible to effectively construct, \revision{also in quadratic time}, a formula in $\withaggr{\FOL}$ that solves $\glbcqa{g()}$.
\end{theorem}
That is, within the class $\withaggr{\sjfbcq}$, we can effectively separate $\AGG$-queries that allow glb rewriting  in $\withaggr{\FOL}$ from those that do not, provided that $\AGG$ satisfies monotonicity and associativity. 
Such a separation result will also be obtained for $\MIN$- and $\MAX$-queries, for both glb and lub (see Theorem~\ref{the:sepminmax}).
\revision{Since the glb rewriting of Theorem~\ref{the:glbmain}, if it exists, can be constructed in quadratic time, its length is at most quadratic.}

As will be argued in Section~\ref{sec:freevariables}, all our results naturally extend to queries with a GROUP BY feature. 
For example, the following SQL query returns, for each dealer, the total quantity of products in stock in their town of operation:
\begin{lstlisting}[
           language=SQL,
           showspaces=false,
           basicstyle=\ttfamily,
           %numbers=left,
           %numberstyle=\tiny,
           commentstyle=\color{gray}
        ]
    SELECT   D.Name, SUM(S.Qty)
    FROM     Dealers AS D, Stock AS S
    WHERE    D.Town = S.Town
    GROUP BY D.Name
\end{lstlisting}
This query is stated as follows in the extended Datalog syntax of~\cite{DBLP:journals/tods/CohenNS06}:
\[
(x,\SUM(y))\leftarrow  \att{Dealers}(\underline{x},t), \att{Stock}(\underline{p,t},y).
\]
Range semantics are obtained by replacing~$x$ with every possible dealer name (``Smith'' and ``James,'' in our example), and calculating the glb and lub as before.

This paper is organized as follows.
Section~\ref{sec:related} discusses related work, and Section~\ref{sec:preliminaries} introduces preliminaries.
Section~\ref{sec:superfrugal} introduces a new class of repairs, called \emph{superfrugal repairs},
which are of standalone interest but will also prove to be a valuable theoretical tool.
In Section~\ref{sec:aggregatelogic}, we formally define $\withaggr{\sjfbcq}$, the class of queries for which we want to compute range consistent query answers. We also define the aggregate logic $\withaggr{\FOL}$ that serves as the target language for our rewritings. Theorem~\ref{the:inexpressible} captures the ``negative'' side of our main separation theorem (Theorem~\ref{the:glbmain}), i.e., the side on which $\glbcqa{g()}$ is not expressible in $\withaggr{\FOL}$.
The ``positive'' side of our separation theorem is captured by Theorem~\ref{the:expressible}, whose proof ideas are sketched in Section~\ref{sec:lb} through an example. 
The full proof, which is technically the most challenging contribution of the paper, is deferred to Appendix~\ref{sec:proofexpressible}.
In Section~\ref{sec:lacking}, we give some insight in the consequences of dropping the assumptions of monotonicity and associativity present in Theorem~\ref{the:glbmain}.
In that section, we also show our separation theorem for $\MIN$- and $\MAX$-queries.
In Section~\ref{sec:unconstrained}, we apply our theoretical findings to disprove a long-standing claim made in~\cite{FuxmanThesis}.
Finally, Section~\ref{sec:discussion} concludes the paper.

\section{Related Work}\label{sec:related}


Consistent query answering (CQA) started by a seminal paper in 1999 co-authored by Arenas, Bertossi, and Chomicki~\cite{DBLP:conf/pods/ArenasBC99}, who introduced the notions of repair and consistent answer. 
Two years later, the same authors introduced the \emph{range semantics} (with lower and upper bounds) for queries with aggregation~\cite{DBLP:conf/icdt/ArenasBC01,DBLP:journals/tcs/ArenasBCHRS03}\cite[Chapter~5]{DBLP:series/synthesis/2011Bertossi}, which has been commonly adopted ever since.
Immediately related to our work is~\cite[Theorem~9]{DBLP:journals/tcs/ArenasBCHRS03}, which establishes that
$\glbcqa{g()}$ is $\NP$-hard (and hence not expressible in $\withaggr{\FOL}$) for $g()\defeq
\COUNTDISTINCT(r)\leftarrow R(\underline{x},r)$, where $R$ is a binary relation whose first attribute is the primary key.
This $\NP$-hardness result is not implied by our Theorem~\ref{the:glbmain}, as $\COUNTDISTINCT$ lacks monotonicity and associativity (see Section~\ref{sec:aggregatingpos} for details).

Ariel Fuxman defined in his PhD thesis~\cite{FuxmanThesis} a syntactic class of self-join-free conjunctive queries, called $\cforest$, whose extension with aggregate operators ($\MAX$, $\MIN$, $\SUM$, $\COUNT$) yields $\caggforest$.  Range semantics for queries in $\caggforest$ can be obtained by executing two first-order queries (one for lower bounds, and one for upper bounds) followed by simple aggregation steps; this technique was subsequently implemented in the ConQuer system~\cite{DBLP:conf/sigmod/FuxmanFM05} through rewriting in SQL.
Although~\cite{FuxmanThesis} permits aggregation over numerical columns containing both positive and negative numbers, as well as zero, we discovered that the SQL rewriting for $\SUM$-queries in~\cite{FuxmanThesis} is flawed when negative numbers are present, as detailed in Section~\ref{sec:unconstrained}.

The class $\cnewclass$~\cite{DBLP:conf/icdt/KhalfiouiW23} is an extension of $\cforest$ that contains all (and only) self-join-free conjunctive queries for which Fuxman's technique applies for~$\COUNT$.
The ease with which $\cforest$ and $\cnewclass$ facilitate range semantics through Fuxman's technique can be attributed to their requirement that joins between non-key and key attributes must involve the entire key of a relation, commonly referred to as ``full'' joins; they do not allow ``partial'' joins  where two relations can join on some (but not all) attributes of a primary key.
In the current paper, we allow partial joins.

Fuxman's technique is different from AggCAvSAT~\cite{DBLP:conf/icde/DixitK22}, a recent system by Dixit and Kolaitis, which uses powerful SAT solvers for computing range semantics, and thus can solve queries that are beyond the computational power of ConQuer.
Aggregation queries were also studied in the context of CQA in~\cite{DBLP:journals/is/BertossiBFL08}.


CQA for self-join-free conjunctive queries~$q$, without aggregation, and primary keys has been intensively studied. 
Its decision variant, which was coined $\cqa{q}$ in 2010~\cite{DBLP:conf/pods/Wijsen10}, asks whether a Boolean query~$q$ is true in every repair of a given database instance.
A systematic study of its complexity for self-join-free conjunctive queries had started already in~2005~\cite{DBLP:conf/icdt/FuxmanM05}, and was eventually solved in two journal articles by Koutris and Wijsen~\cite{DBLP:journals/tods/KoutrisW17,DBLP:journals/mst/KoutrisW21}, as follows: for every self-join-free Boolean conjunctive query~$q$, $\cqa{q}$ is either in $\FO$, $\LOGSPACE$-complete, or $\coNP$-complete, and it is decidable, given~$q$, which case applies.
This complexity classification extends to non-Boolean queries by treating free variables as constants.
Other extensions beyond this trichotomy deal with foreign keys~\cite{DBLP:conf/pods/HannulaW22}, more than one key per relation~\cite{DBLP:conf/pods/KoutrisW20}, negated atoms~\cite{DBLP:conf/pods/KoutrisW18}, or restricted self-joins~\cite{DBLP:conf/pods/KoutrisOW21,DBLP:journals/pacmmod/KoutrisOW24,DBLP:journals/pacmmod/PadmanabhaSS24}.
For unions of conjunctive queries~$q$, Fontaine~\cite{DBLP:journals/tocl/Fontaine15} established interesting relationships between $\cqa{q}$  and Bulatov's dichotomy theorem
for conservative CSP~\cite{DBLP:journals/tocl/Bulatov11}.
\revision{Recently, Figueira et al.~\cite{DBLP:conf/icdt/FigueiraPSS23} proposed a polynomial-time algorithm for solving $\cqa{q}$ for conjunctive queries~$q$, including those with self-joins, and showed that it can replace all polynomial-time algorithms found in~\cite{DBLP:journals/tods/KoutrisW17,DBLP:conf/pods/KoutrisOW21}.}

The counting variant $\sharp\cqa{q}$ asks to count the number of repairs that satisfy some Boolean query~$q$.
This counting problem is fundamentally different from the range semantics in the current paper.
For self-join-free conjunctive queries, $\sharp\cqa{q}$ exhibits a dichotomy between  $\FP$ and $\sharp\PTIME$-complete under polynomial-time Turing reductions~\cite{DBLP:journals/jcss/MaslowskiW13}. This dichotomy has been shown to extend to queries with self-joins if primary keys are singletons~\cite{DBLP:conf/icdt/MaslowskiW14}, and to functional dependencies~\cite{DBLP:conf/pods/CalauttiLPS22a}.
Calautti, Console, and Pieris present in~\cite{DBLP:conf/pods/CalauttiCP19} a complexity analysis of these counting problems under weaker reductions, in particular, under many-one logspace reductions.
The same authors have conducted an experimental evaluation of randomized approximation schemes for approximating the percentage of repairs that satisfy a given query~\cite{DBLP:conf/pods/CalauttiCP21}.
Other approaches to making CQA more meaningful and/or tractable include operational repairs~\cite{DBLP:conf/pods/CalauttiLP18,DBLP:conf/pods/CalauttiLPS22} and preferred repairs~\cite{DBLP:journals/tcs/KimelfeldLP20,DBLP:journals/amai/StaworkoCM12}.

Recent overviews of two decades of theoretical research in CQA can be found in~\cite{DBLP:conf/pods/Bertossi19,DBLP:journals/sigmod/Wijsen19,DBLP:journals/cacm/KimelfeldK24}.
It is worthwhile to note that theoretical research in $\cqa{q}$ has stimulated implementations and experiments in prototype systems~\cite{DBLP:conf/sat/DixitK19,DBLP:journals/pacmmod/FanKOW23,DBLP:conf/sigmod/FuxmanFM05,DBLP:conf/vldb/FuxmanFM05,DBLP:conf/cikm/KhalfiouiJLSW20,DBLP:journals/pvldb/KolaitisPT13}.

\section{Preliminaries}\label{sec:preliminaries}

We assume denumerable sets $\var$ and $\dom$ of variables and constants respectively.
The set $\dom$ includes $\rationals_{\geq 0}$, the set of non-negative rational numbers.
The set of \emph{numerical variables} is a subset of $\var$.

We assume denumerably many relation names. Every relation name is associated with a \emph{signature}, which is a triple $(n,k,J)$ where $n$ is the \emph{arity}, $\{1,\ldots,k\}$ is called the \emph{primary key}, and $J\subseteq\{1,\ldots,n\}$ are \emph{numerical positions} (also called \emph{numerical columns}). 
This relation name is \emph{full-key} if $n=k$.
\revision{Note that each relation name~$R$ is associated with exactly one key constraint, which is determined by the signature of~$R$.} For an $n$-tuple $\vec{x}=(x_{1},\ldots,x_{n})$, we write~$\card{\vec{x}}$ to denote its \emph{arity}~$n$.
We often blur the distinction between a sequence $(x_{1},\ldots,x_{n})$ of distinct variables and the set $\{x_{1},\ldots,x_{n}\}$, which is also denoted $\vars{\vec{x}}$.

\paragraph{\bf Atoms, facts, and database instances}
Let $R$ be a relation name of signature  $(n,k,J)$.
An \emph{atom} is an expression $R(u_{1},\ldots,u_{n})$ where each $u_{i}$ is either a constant or a variable, and for every $j\in J$, $u_{j}$~is a numerical variable or a number in~$\rationals_{\geq 0}$. 
It is common to underline positions of the primary key.
If $F$ is an atom, then $\vars{F}$ is the set of variables that occur in~$F$, and $\key{F}$ is the set of variables that occur in $F$ at a position of the primary key. Further, we define $\notkey{F}\defeq\vars{F}\setminus\key{F}$.
A \emph{fact} is an atom without variables.
A fact with relation name~$R$ is also called an $R$-fact.
Two facts $R_{1}(\underline{\vec{a}_{1}}, \vec{c}_{1})$ and $R_{2}(\underline{\vec{a}_{2}}, \vec{c}_{2})$ are said to be \emph{key-equal} if $R_{1}=R_{2}$ and $\vec{a}_{1}=\vec{a}_{2}$.

A \emph{database instance} is a finite set of facts.
If $R$ is a relation name, then the \emph{$R$-relation of $\db$} is the set of all $R$-facts in $\db$.
A \emph{database instance} is \emph{consistent} if it does not contain two distinct facts that are key-equal.
Please be aware that there is no requirement to explicitly specify the primary keys, as they are derived from the predetermined signatures associated with the relation names of the facts present in the database instance.

\paragraph{\bf Valuation.}
A \emph{valuation} over a finite set $U$ of variables is a total mapping $\theta$ from $U$ to~$\dom$ such that $\theta(r)\in\rationals_{\geq 0}$ for every numeric variable~$r$.
For a valuation $\theta$ over~$U$, we write $\domain{\theta}$ to denoted its \emph{domain}~$U$.
A valuation $\theta$ over $U$ is extended to every element $u$ in $\var\cup\dom$ by letting $\theta(u)=u$ for every $u\notin U$.
If $\theta$ is a valuation and $V\subseteq\domain{\theta}$, then $\restrict{\theta}{V}$ denotes the restriction of $\theta$ to $V$, i.e., $\domain{\restrict{\theta}{V}}=V$ and for every $x\in V$, we have $\restrict{\theta}{V}(x)=\theta(x)$.

Let $\theta$ be a valuation.
If~$F$ is the atom $R(u_{1},\ldots,u_{n})$, then $\theta(F)\defeq R(\theta(u_{1}),\ldots,\theta(u_{n}))$.
If~$q$ is a set of atoms, then $\theta(q)\defeq\{\theta(F)\mid F\in q\}$.
The set $\theta(q)$ is also denoted $\substitute{q}{\vec{y}}{\vec{c}}$ where $\vec{y}$ is a shortest sequence containing every variable in $\domain{\theta}$, and $\vec{c}\defeq\theta(\vec{y})$.
Significantly, every variable in $\vars{q}\setminus\domain{\theta}$ remains a variable in $\theta(q)$.

\paragraph{\bf Partial valuation.}
Let $\db$ be a database instance, and $\varphi(\vec{x})$ a first-order formula with free variables $\vec{x}$. 
Let $\theta$ be a valuation. 
Then we write $(\db,\theta)\models\varphi(\vec{x})$ to denote that $\theta$ can be extended to a valuation~$\theta'$ over $\domain{\theta}\cup\vars{\vec{x}}$ such that for $\vec{a}\defeq\theta'(\vec{x})$, we have $\db\models\varphi(\vec{a})$ using standard semantics (see, e.g., \cite[p. 15]{DBLP:books/sp/Libkin04}).
Typically, but not necessarily, $\domain{\theta}\subseteq\vars{\vec{x}}$.
If $\varphi$ has no free variables, then we write $\db\models\varphi$ instead of $(\db,\emptyvaluation)\models\varphi$, where $\emptyvaluation$ is the empty valuation.

If $\varphi(x_{1},\ldots,x_{n})$ is a formula with $n$~distinct free variables $x_{1},\ldots,x_{n}$, and $y_{1},\ldots,y_{n}$ are $n$~distinct variables (not necessarily disjoint with the $x_{i}$s), then $\varphi(y_{1},\ldots,y_{n})$ denotes the formula obtained from $\varphi(x_{1},\ldots,x_{n})$ by replacing, for each $1\leq i\leq n$, every free occurrence of $x_{i}$ by $y_{i}$.

\paragraph{\bf Repairs and $\cqa{\varphi}$}
A \emph{repair} of a database instance is a $\subseteq$-maximal consistent subset of it.
We write $\repairs{\db}$ for the set of repairs of a database instance~$\db$.
If $\varphi(\vec{x})$ is a first-order formula and $\theta$ a valuation, then we write $(\db,\theta)\cqamodels\varphi(\vec{x})$ to denote that for every repair~$\rep$ of $\db$, we have $(\rep,\theta)\models\varphi(\vec{x})$.
If $\varphi$ has no free variables, then we write $\db\cqamodels\varphi$ instead of $(\db,\emptyvaluation)\cqamodels\varphi$, where~$\emptyvaluation$ is the empty valuation.
For a closed formula $\varphi$, $\cqa{\varphi}$ is the decision problem that takes a database instance $\db$ as input and determines whether or not $\db\cqamodels\varphi$.

\paragraph{\bf Boolean Conjunctive Queries}
A \emph{self-join-free Boolean conjunctive query}~$q$ is a closed first-order formula
$\exists\vec{u}\lrformula{R_{1}(\vec{u}_{1})\land\dotsm\land R_{n}(\vec{u}_{n})}$,
where each $R_{i}(\vec{u}_{i})$ is an atom, $\vec{u}$ is a sequence containing every variable occurring in some $\vec{u}_{i}$, and $i\neq j$ implies $R_{i}\neq R_{j}$.
The conjunction $R_{1}(\vec{u}_{1})\land\dotsm\land R_{n}(\vec{u}_{n})$, whose free variables are $\vec{u}$, is called the \emph{body} of the query.
We write $\sjfbcq$ for the set of self-join-free Boolean conjunctive queries.
We often blur the distinction between the Boolean query~$q$, its body, and the set $\{R_{1}(\vec{u}_{1}),\ldots,R_{n}(\vec{u}_{n})\}$.
For example, if $F$ is an atom of (the body of)~$q$, then $q\setminus\{F\}$ is the query obtained from $q$ by deleting $F$ from its body. 

\begin{figure}
    \centering
    \begin{tabular}{c|c}
    \includegraphics[scale=0.8]{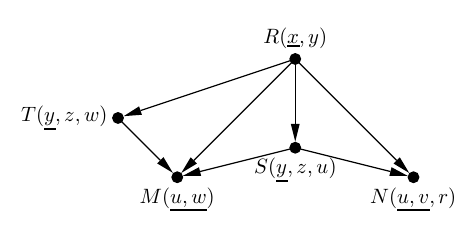}
    &
    \includegraphics[scale=0.8]{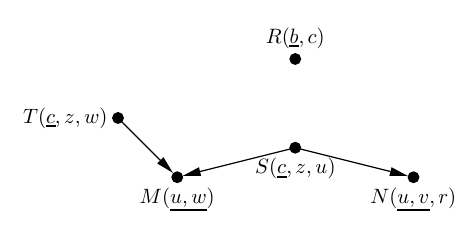}
    \\
    $q_{0}$ & $\substitute{q_{0}}{xy}{bc}$
    \end{tabular}
    \caption{Attack graphs for two queries in $\sjfbcq$: The query on the right is derived from the query on the left by initializing $x$ to $b$ and $y$ to $c$.}
    \label{fig:notvarConnected}
\end{figure}

\paragraph{\bf Attack graph}
Attack graphs of queries in $\sjfbcq$ were first introduced in~\cite{DBLP:journals/tods/Wijsen12} and subsequently generalized in~\cite{DBLP:journals/tods/KoutrisW17}.
Their definition is technical and not immediately intuitive.
However, as will become clear in the technical treatment, the partial order induced by acyclic attack graphs will play an essential role in query rewriting.

Let $q\in\sjfbcq$.
We write $\fdset{q}$ for the set of functional dependencies that contains $\fd{\key{F}}{\vars{F}}$ whenever $F\in q$.
For $F\in q$, we define $\keycl{F}{q}\defeq\{x\in\vars{q}\mid\fdset{q\setminus\{F\}}\models\fd{\key{F}}{x}\}$, where $\models$ is the standard notion of logical implication.
An atom $F$ of $q$ is said to \emph{attack} a variable~$x$, denoted $F\attacks{q}x$, if there is a sequence of variables $(x_{1},x_{2},\ldots,x_{m})$ such that no $x_{i}$ is in $\keycl{F}{q}$, $x_{1}\in\notkey{F}$, $x_{m}=x$, and every two adjacent variables occur together in some atom of~$q$. 
A variable $x$ is said to be \emph{unattacked (in~$q$)} if no atom attacks~$x$.
The \emph{attack graph} of~$q$ is a directed simple graph whose vertices are the atoms of~$q$. There is a directed edge from $F$ to $G$, denoted $F\attacks{q}G$, if $F$ attacks some variable of $\vars{G}$. 

Whenever a query in $\sjfbcq$ is clear from the context, we can use a relation name as a shorthand for the unique atom with that relation name in the query. For example, in the following example, $R$~is a shorthand for the atom $R(\underline{x},y)$.

\begin{example}
    Consider the query~$q_{0}$ in $\sjfbcq$ whose atoms are shown on the left side of~Fig.~\ref{fig:notvarConnected}.
    The directed edges represent attacks.
    We have $\keycl{R}{q_{0}}=\{x\}$, $\keycl{T}{q_{0}}=\{y,z,u\}$, $\keycl{S}{q_{0}}=\{y,z,w\}$, $\keycl{M}{q_{0}}=\{u,w\}$, and $\keycl{N}{q_{0}}=\{u,v\}$. 
    The sequence $(y,u)$, for example, implies that $R\attacks{q_{0}}M$ and  $R\attacks{q_{0}}N$.
    The attack graph of $\substitute{q_{0}}{xy}{bc}$ is shown on the right side of~Fig.~\ref{fig:notvarConnected}. It is generally true that an acyclic attack graph remains acyclic if some variables are initialized.
 \myqed   
\end{example}

The following result shows the usefulness of attack graphs.

\begin{theorem}[\cite{DBLP:journals/tods/KoutrisW17}]\label{the:forewritability}
For every query $q$ in $\sjfbcq$, $\cqa{q}$ is in $\FO$ if and only if the attack graph of $q$ is acyclic.
\end{theorem}


\section{$\forall$Embeddings and Superfrugal Repairs}\label{sec:superfrugal}

In~\cite{DBLP:journals/tods/KoutrisW17}, a construct called ``frugal repair'' was introduced. In the current section, we introduce a related but more stringent construct, called \emph{superfrugal repair}.

\paragraph{\bf Embeddings and $\forall$embeddings}
Whenever $q(\vec{u})$ is a conjunction of atoms, we write $q$ to denote the closed formula $\exists\vec{u}\formula{q(\vec{u})}$.
Let $q(\vec{u})$ now be a self-join-free conjunction of $n$~atoms such that the attack graph of $q$ is acyclic.
The following definitions are relative to a fixed topological sort  $(F_{1},\ldots,F_{n})$ of $q$'s attack graph and a fixed database instance~$\db$.
We define the following sequences of variables for $\ell\in\{1,\ldots,n\}$:
\begin{itemize}
    \item $\vec{u}_{\ell}$ contains all (and only) variables of $\bigcup_{i=1}^{\ell}\vars{F_{i}}$. Thus, $\vec{u}_{n}=\vec{u}$;
    \item $\vec{x}_{\ell}$ contains the variables of $\key{F_{\ell}}$ that do not already occur in $\bigcup_{i=1}^{\ell-1}\vars{F_{i}}$; and
    \item $\vec{y}_{\ell}$  contains the variables of $\notkey{F_{\ell}}$ that do not already occur in $\bigcup_{i=1}^{\ell-1}\vars{F_{i}}$.
\end{itemize}
Moreover, we define $\vec{u}_{0}=()$, the empty sequence. 
With this notation, we have that for every $\ell\in\{1,\ldots,n\}$, 
\begin{equation}\label{eq:uxy}
\vec{u}_{\ell}=(\vec{u}_{\ell-1}, \vec{x}_{\ell}, \vec{y}_{\ell}).
\end{equation}

Let $\ell\in\{1,\ldots,n\}$ in what follows.
An \emph{$\ell$-embedding (of $q$ in $\db$)} is a valuation $\theta$ over~$\vec{u}_{\ell}$ such that $(\db,\theta)\models q(\vec{u})$. 
Further, a $0$-embedding is defined to be the empty set.
An $\ell$-embedding with $\ell=n$ is also called an \emph{embedding} for short.
\revision{A set $M$ of embeddings is said to be \emph{consistent} if $M\models\fdset{q}$.}

The following definition is by induction.
An $\ell$-embedding $\theta$ is called an \emph{$\ell$-$\forall$embedding (of~$q$ in~$\db$)} if one of the following holds true:
\begin{description}
    \item[Basis:] $\ell=0$ and every repair of $\db$ satisfies~$q$; or
    \item[Step:] $\ell\geq 1$ and both the following hold true:
    \begin{itemize}
    \item 
        $(\db,\restrict{\theta}{\vec{u}_{\ell-1}\vec{x}_{\ell}})\cqamodels F_{\ell}\land F_{\ell+1}\land \dotsm\land F_{n}$; and 
    \item 
    the $(\ell-1)$-embedding contained in $\theta$ is an $(\ell-1)$-$\forall$embedding.
    \end{itemize}
\end{description}
In simple terms, the first bullet states that every repair must satisfy the query whose atoms are obtained from $F_{\ell},F_{\ell+1},\ldots,F_{n}$ by replacing $x$ with $\theta(x)$ whenever $x$ is a variable that occurs in the primary key of $F_{\ell}$ or in an atom that precedes $F_{\ell}$.
The second bullet implies that the same condition must hold for $\ell-1, \ell-2, \ldots, 2, 1$, and eventually~$0$.
An $\ell$-$\forall$embedding with $\ell=n$ is also called a \emph{$\forall$embedding} for short.




\begin{example}\label{ex:forallembedding}
The query $q_{0}=\exists t\exists p\formula{\att{Dealers}(\underline{\textnormal{``James''}},t)\land \att{Stock}(\underline{p,t},\textnormal{35})}$ checks if there is any product stored in a quantity of~$35$ in the town where James is a dealer.
It holds true in every repair of the database instance~$\dbstock$ of Fig.~\ref{fig:dealers}.
The embedding $\{t\mapsto\textnormal{``Boston''}, p\mapsto\textnormal{``Tesla~Y''}\}$ is a $\forall$embedding.
On the other hand, the embedding $\theta\defeq\{t\mapsto\textnormal{``Boston''}, p\mapsto\textnormal{``Tesla~X''}\}$ is not a $\forall$embedding, because $\restrict{\theta}{\{t,p\}}=\theta$ and $(\dbstock,\theta)\not\cqamodels \att{Stock}(\underline{p,t},\textnormal{35})$.
Indeed, if $\rep$ is a repair that contains $\att{Stock}(\underline{\textnormal{``Tesla~X''},\textnormal{``Boston''}},\textnormal{40})$,
then $(\rep,\theta)\not\models\att{Stock}(\underline{p,t},\textnormal{35})$.
\myqed
\end{example}

We now state two important helping lemmas. The first one states that all topological sorts of an acyclic attack graph yield the same $\forall$embeddings. The second lemma establishes that  $\forall$embeddings can be computed in $\FOL$.

\begin{lemma}\label{lem:anysort}
Let $q$ a query in $\sjfbcq$ with an acyclic attack graph, and $\db$ be a database instance.
Let $n$ be the number of atoms in~$q$.
Let $\sort_{1}$ and $\sort_{2}$ be two topological sorts of $q$'s attack graph.
Every $n$-$\forall$embedding relative to~$\sort_{1}$ is an $n$-$\forall$embedding relative to $\sort_{2}$.
\end{lemma}

\begin{lemma}\label{lem:allfo}
Let $q\defeq\exists\vec{u}\formula{q(\vec{u})}$ be a query in $\sjfbcq$ with an acyclic attack graph.
It is possible to construct, \revision{in quadratic time in the size of~$q$}, a $\FOL$ formula $\varphi(\vec{u})$ such that for every database instance~$\db$, for every valuation $\theta$ over $\vec{u}$,
$(\db,\theta)\models\varphi(\vec{u})$ if and only if $\theta$ is a $\forall$embedding of~$q$ in~$\db$. 
\end{lemma}

\revision{Note that since the formula $\varphi(\vec{u})$ in Lemma~\ref{lem:allfo} can be constructed in quadratic time, its length is at most quadratic (in the size of~$q$).
}

\paragraph{\bf Superfrugal repairs}
A repair~$\rep$ of a database instance~$\db$ is \emph{superfrugal} relative to $\exists\vec{u}\formula{q(\vec{u})}$ if every embedding of~$q$ in $\rep$ is a $\forall$embedding of~$q$ in~$\db$.
Informally, superfrugal repairs are repairs with $\subseteq$-minimal sets of embeddings, which is expressed by Lemma~\ref{lem:frugal}.

\begin{example}\label{ex:frugal}
Continuation of Example~\ref{ex:forallembedding}.
Let $\rep$ be the repair of $\dbstock$ that contains all (and only) tuples preceded by $\dagger$ in Fig.~\ref{fig:dealers}.
Then, $\rep$ is not superfrugal relative to~$q_{0}\defeq\exists t\exists p\formula{\att{Dealers}(\underline{\textnormal{``James''}},t)\land \att{Stock}(\underline{p,t},\textnormal{35})}$.
Indeed, $\theta\defeq\{t\mapsto\textnormal{``Boston''}, p\mapsto\textnormal{``Tesla~X''}\}$ is an embedding of $q_{0}$ in~$\rep$, but as discussed in Example~\ref{ex:forallembedding}, $\theta$ is not a $\forall$embedding of $q_{0}$ in~$\db$.
\myqed
\end{example}


\begin{lemma}\label{lem:frugal}
Let $\db$ be a database instance, and $\exists\vec{u}\formula{q(\vec{u})}$ a query in $\sjfbcq$ with an acyclic attack graph.
For every repair $\rep$ of~$\db$, there exists a superfrugal repair $\rep^{*}$ of $\db$ such that for every sequence~$\vec{c}$ of constants, of length $\card{\vec{u}}$,  if $\rep^{*}\models q(\vec{c})$, then $\rep\models q(\vec{c})$.
\end{lemma}

\revision{
Lemma~\ref{lem:minimal} in Appendix~\ref{sec:minimality} shows that for queries~$q\in\sjfbcq$ with an acyclic attack graph, superfrugal repairs are identical to the $n$-minimal repairs defined in~\cite{DBLP:conf/icdt/FigueiraPSS23} and the $\frugal{X}{q}$-frugal repairs introduced in~\cite{DBLP:journals/tods/KoutrisW17},  where $n$ denotes the number of atoms and $X=\vars{q}$. 
In the current paper, we opted to define superfrugal repairs in terms of $\forall$embeddings; an alternative approach would be to take $n$-minimal repairs (or, equivalently, $\preceq_{q}^{X}$-frugal repairs) as a starting point and show that all embeddings in them are $\forall$embeddings.  
This alternative approach would also arrive at the conclusion of Lemma~\ref{lem:allfo} regarding the computability of $\forall$embeddings in $\FOL$.
Specifically, Lemma~8 and Remark~11 in~\cite[Section~4]{DBLP:conf/icdt/FigueiraPSS23} entail that one can compute in $\FOL$ a set~$E$ that contains, for every $\ell\in\{0,1,2,\ldots,n\}$, each $\ell$-$\forall$embedding. 
The goal in~\cite{DBLP:conf/icdt/FigueiraPSS23} is to check whether $E$ contains a $0$-$\forall$embedding (or, equivalently, the empty set), which indicates that~$q$ holds true in every repair. 
For that purpose, it is sufficient for the set $E$ to be a superset, rather than an exact match, of the set of all $\ell$-$\forall$embeddings, for $0\leq\ell\leq n$.
Lemma~\ref{lem:allfo} differs in that it refers to a set that contains exactly all $n$-$\forall$embeddings and nothing else.
On the other hand, we will argue in Section~\ref{sec:discussion} that the more general technical development in~\cite{DBLP:conf/icdt/FigueiraPSS23}, which also handles cyclic attack graphs, could be particularly valuable for exploring the polynomial-time computability of $\glbcqa{g()}$, an intriguing open problem that is not the focus of our current submission.
}

\section{Aggregate Logic and CQA}\label{sec:aggregatelogic}

In this section, we first define the notion of \emph{aggregate operator} and then introduce the logic $\withaggr{\FOL}$ which will serve as the target language for our rewritings.
We formally define consistent lower and upper bounds, denoted $\lubcqa{g(\vec{x})}$ and $\glbcqa{g(\vec{x})}$, for arbitrary numerical terms $g(\vec{x})$ in $\withaggr{\FOL}$ with free variables~$\vec{x}$.
Our study will subsequently focus on closed numerical terms~$g()$ in~$\withaggr{\sjfbcq}$, a subclass of  $\withaggr{\FOL}$ introduced in Section~\ref{sec:problemstatement}.

\subsection{Aggregating Non-Negative Numbers}\label{sec:aggregatingpos}

A \emph{(positive) aggregate operator} is a function $\calF$ that takes as argument a finite multiset $X$ of non-negative rational numbers such that $\calF(X)\in\rationals_{\geq 0}$ if $X\neq\emptyset$, and $\calF(\emptyset)=f_{0}$.
Here, $f_{0}$ is a constant (not necessarily in~$\rationals_{\geq 0}$), which is the value returned by $\calF$ on the empty multiset.
In the following definitions, all multisets are understood to be multisets of non-negative rational numbers.

An aggregate operator is \emph{associative} if for all non-empty multisets~$X$ and $Y$ such that $X\neq\emptyset$, we have $\calF(X\uplus Y)=\calF(\bag{\calF(X)}\uplus Y)$, where $\uplus$ denotes union of multisets.

\begin{example}
Examples of associative aggregate operators are $\agmax$, $\agmin$, and $\agsum$.
Examples of aggregate operators that are not  associative are $\agavg$, $\agcount$, and $\agsumdistinct$.
Note, for instance, $\agcount(\bag{5,6,7,8})=4$ and $\agcount(\bag{\agcount(\bag{5,6,7}),8})=\agcount(\bag{3,8})=2$. 
\myqed
\end{example}

An aggregate operator~$\calF$ is \emph{monotone} if for all $m>0$ and every (possibly empty) multiset $Y$, we have $\calF(\bag{x_{1},\ldots,x_{m}})\leq \calF(\bag{ x_{1}',\ldots,x_{m}'}\uplus Y)$ whenever $x_{i}\leq x_{i}'$ for every~$i$.

\begin{example}
Examples of monotone aggregate operators are $\agmax$, $\agsum$, and $\agcount$.
Note that $\agmin$ is not monotone since $\agmin(\bag{3})>\agmin(\bag{2,3})$.
$\COUNTDISTINCT$ also lacks monotonicity: if we increase $3$ to $4$ in the multiset $\bag{3,4}$, the number returned by $\COUNTDISTINCT$ drops from~$2$ to~$1$.  
\myqed
\end{example}

Significantly, an aggregate operator $\calF$ that is not monotone in general may become monotone under a restriction of its underlying domain. For example, $\agprod$, defined by $\agprod(\bag{x_{1},\ldots,x_{m}})\defeq\Pi_{i=1}^{m}x_{i}$, is not monotone over $\rationals_{\geq 0}$ (because $\agprod(\bag{1})>\agprod(\bag{1,\frac{1}{2}})$), but is monotone over $\rationals_{\geq 1}$.


\subsection{The Logic $\withaggr{\FOL}$}

Our treatment of aggregate logic follows the approach in~\cite{DBLP:journals/jacm/HellaLNW01,DBLP:books/sp/Libkin04}.
We write $\withaggr{\FOL}$ for the extension of predicate calculus introduced next. 

Let $q(\vec{x},\vec{y})$ be a formula, where $\vec{x}$, $\vec{y}$ are disjoint sequences that together contain each free variable of~$q$ exactly once.
A \emph{primitive numerical term} is either a non-negative rational number or a numerical variable in $\vec{x}\vec{y}$. 
\revision{For example, the query of Example~\ref{ex:forallembedding} uses the constant numerical term~$35$.}
For every possible aggregate operator $\calF$, a primitive numerical term~$r$, and a formula $q(\vec{x},\vec{y})$, we have a new \emph{numerical term} 
\begin{equation*}\label{eq:gxnt}
g(\vec{x})=\agterm{\calF}{\vec{y}}{r}{q(\vec{x},\vec{y})}.
\end{equation*}
Variables $\vec{y}$ that are free in $q(\vec{x},\vec{y})$ become bound in $\agterm{\calF}{\vec{y}}{r}{q(\vec{x},\vec{y})}$;
in this respect $\agtermquantification{\calF}{\vec{y}}$ behaves like a sort of quantification over~$\vec{y}$.
Next, we define the semantics.

Let $\vec{a}$ be a sequence of constants of length~$\card{\vec{x}}$.
The value $g(\vec{a})$ on a database instance $\db$, denoted $\interpret{g(\vec{a})}{\db}$, is calculated as follows.
If there is no~$\vec{b}$ such that $\db\models q(\vec{a},\vec{b})$, then $\interpret{g(\vec{a})}{\db}=f_{0}$, with $f_{0}$ as defined in Section~\ref{sec:aggregatingpos}.
Otherwise, let $\theta_{1},\ldots,\theta_{m}$ enumerate (without duplicates) all valuations over $\vars{\vec{x}\vec{y}}$ such that $\theta_{i}(\vec{x})=\vec{a}$ and $(\db,\theta_{i})\models q(\vec{x},\vec{y})$ for $1\leq i\leq m$.
Then, $\interpret{g(\vec{a})}{\db}=\calF(\leftbag\theta_{1}(r),\ldots,\theta_{m}(r)\rightbag)$.
Note that the argument of~$\calF$ is in general a multiset, since $\theta_{i}(r)$ may be equal to $\theta_{j}(r)$ for $i\neq j$.
A numerical term without free variables is also called a \emph{numerical query}, denoted $g()$.

\begin{example}\label{ex:nested}
    Retrieve the town(s) with the largest total quantity of stored products.
    First, we provide a formula $q_{0}(t,y)$ which holds true if $y$ is the total quantity of products stored in town~$t$:
    \[q_{0}(t,y)\defeq\exists p\exists z\lrformula{\att{Stock}(\underline{p,t},z)}\land y=\agterm{\agsum}{(p,z)}{z}{\att{Stock}(\underline{p,t},z)}.\]
    The final query is 
    $\exists y\lrformula{q_{0}(t,y)\land y=\agterm{\agmax}{(t,z)}{z}{q_{0}(t,z)}}$,
    with $q_{0}$ as previously defined.
    The variable $z$ is used for clarity, but can be renamed in~$y$ without creating a naming conflict.
 \myqed   
\end{example}

\subsection{Lower and Upper Bounds Across All Repairs}

For each numerical term $g(\vec{x})\defeq\agterm{\calF}{\vec{y}}{r}{q(\vec{x},\vec{y})}$, we define $\glbcqa{g(\vec{x})}$ and $\lubcqa{g(\vec{x})}$ relative to a database instance~$\db$.
Let $\vec{a}$ be a sequence of constants of length~$\card{\vec{x}}$.
If there is a repair~$\rep$ of $\db$ such that $\interpret{g(\vec{a})}{\rep}=f_{0}$, then $\interpret{\glbcqa{g(\vec{a})}}{\db}=\bot$; otherwise 
\[\interpret{\glbcqa{g(\vec{a})}}{\db}\defeq\min\left\{\interpret{g(\vec{a})}{\rep}\mid\rep\in\repairs{\db}\right\}.\] 
The value $\interpret{\lubcqa{g(\vec{a})}}{\db}$ is defined symmetrically by replacing $\min$ with $\max$. 
Note that $\bot$ is returned if (and only if) $\vec{a}$ is not a consistent answer to the query $\{\vec{x}\mid \exists \vec{y}\formula{q(\vec{x},\vec{y})}\}$, a situation that we typically want to distinguish in CQA. Through the introduction of $\bot$, our glb and lub operations become independent from any conventions regarding aggregation over the empty multiset.


\subsection{Range CQA for $\withaggr{\sjfbcq}$}\label{sec:problemstatement}

Our focus will be on glb and lub rewriting  of numerical terms 
$g(\vec{x})\defeq\agterm{\calF}{\vec{y}}{r}{q(\vec{x},\vec{y})}$ 
where $q(\vec{x},\vec{y})$ is a self-join-free conjunction of atoms.
For readability, we often express such a numerical term $g(\vec{x})$ using the following Datalog-like syntax:
\[
\formula{\vec{x},\AGG(r)}\leftarrow q(\vec{x},\vec{y}),
\]
where the aggregate symbol $\AGG$ is interpreted by the aggregate operator~$\calF$, which is often made explicit by writing $\agagg$.
For instance, $\SUM$ is interpreted by $\agsum$, and $\MAX$ by $\agmax$.

Our research question is: Under which conditions can $\glbcqa{g(\vec{x})}$ and $\lubcqa{g(\vec{x})}$ be expressed in $\withaggr{\FOL}$?
In the initial technical treatment, we assume that $\vec{x}$ is empty, i.e., we are dealing with numerical terms $g()\defeq\agterm{\calF}{\vec{y}}{r}{q(\vec{y})}$ without free variables, which are conveniently expressed as $\AGG(r)\leftarrow q(\vec{y})$.
In Section~\ref{sec:freevariables}, we treat the extension to free variables.
The following definition pinpoints the class of queries we are interested in; it is followed by a first inexpressibility result, which captures the ``no''-side of Theorem~\ref{the:glbmain}.

\begin{definition}
    $\withaggr{\sjfbcq}$ is defined as the class of numerical queries $\agterm{\calF}{\vec{y}}{r}{q(\vec{y})}$ where~$q(\vec{y})$ is a self-join-free conjunction of atoms. 
    An alternative syntax is $\AGG(r)\leftarrow q(\vec{y})$, where the aggregate symbol $\AGG$ is interpreted by $\calF$ (which is often made explicit by writing $\agagg$ instead of~$\calF$). 
\myqed    
\end{definition}

\begin{theorem}\label{the:inexpressible}
Let $g()\defeq\AGG(r)\leftarrow q(\vec{y})$ be a numerical query in $\withaggr{\sjfbcq}$.
If the attack graph of $\exists\vec{y}\formula{q(\vec{y})}$ is cyclic, then $\glbcqa{{g}()}$ is not expressible in $\withaggr{\FOL}$.
\end{theorem}

\revision{The crux of the proof of Theorem~\ref{the:inexpressible} is a deep result by Hella et al.~\cite{DBLP:journals/jacm/HellaLNW01}, which implies that every query in $\withaggr{\FOL}$ is Hanf-local; see also~\cite[Corollary~8.26 and Exercise~8.16]{DBLP:books/sp/Libkin04}.
It follows from~\cite{DBLP:journals/tods/KoutrisW17} that $\glbcqa{{g}()}$ is not Hanf-local if the underlying attack graph is cyclic.
}
Moreover, from the proof of Theorem~\ref{the:inexpressible}, it becomes immediately clear that the theorem remains valid if we replace $\glbcqa{q}$ with $\lubcqa{{g}()}$.

\section{Aggregate Operators that are Both Associative and Monotone}\label{sec:lb}

In this section, we show that the inverse of Theorem~\ref{the:inexpressible} holds under the restriction that aggregate operators are monotone and associative.
Later, in Section~\ref{sec:lacking}, we will show that the inverse does not hold without this restriction.

\begin{theorem}\label{the:expressible}
 Let $g()\defeq\AGG(r)\leftarrow q(\vec{y})$ be a numerical query in $\withaggr{\sjfbcq}$ such that~$\agagg$ is monotone and associative, and the attack graph of $\exists\vec{y}\formula{q(\vec{y})}$ is acyclic.
      Then, $\glbcqa{{g}()}$ is expressible in $\withaggr{\FOL}$.
\end{theorem}
$\SUM$ and $\MAX$ are probably the most common aggregate operators that are monotone and associative, and are thus covered by Theorem~\ref{the:expressible}. Note that $\COUNT$-queries are also covered, because they can be written in the form $\SUM(1)\leftarrow q(\vec{y})$, \revision{where $1$ is a constant numerical term}.  

\revision{
It follows from Theorems~\ref{the:forewritability}, \ref{the:inexpressible}, and~\ref{the:expressible} that for a numerical query $g()\defeq\AGG(r)\leftarrow q(\vec{y})$ in $\withaggr{\sjfbcq}$ such that $\agagg$ is both monotone and associative,
if $\cqa{\exists y\formula{q(\vec{y})}}$ is expressible in $\FOL$, then $\glbcqa{{g}()}$ is expressible in $\withaggr{\FOL}$ (and the inverse direction is easily seen to hold as well).
However, it will soon become apparent that transitioning from $\cqa{\exists y \formula{q(\vec{y})}}$ to $\glbcqa{g()}$ introduces new challenges that require novel techniques. Briefly, in the former problem, it is sufficient to determine the existence (or non-existence) of a $0$-$\forall$embedding. In contrast, the latter problem, if a $0$-$\forall$embedding exists, additionally requires finding a $\subseteq$-maximal consistent set of $\forall$embeddings whose aggregated value is minimal.
}
The detailed proof of Theorem~\ref{the:expressible} in Appendix~\ref{sec:proofexpressible} consists the main challenge of this paper.
In the remainder of this section, we use a concrete example to introduce the main ingredients of that proof. 
Together with Theorem~\ref{the:inexpressible}, this will imply our main result, Theorem~\ref{the:glbmain}, whose proof is given in Appendix~\ref{sec:proofglbmain}. 

\subsection{Main Ideas for Proving Theorem~\ref{the:expressible}}

\begin{figure}\centering
\begin{small}
\[
\begin{array}{ccc}
\begin{array}{c|*{2}{c}}
 {R} & \underline{x} & {y}\bigstrut\\\cline{2-3}
          & a_1 & b_1\\
          & a_1 & b_2\\\cdashline{2-3}
          & a_2 & b_2\\
          & a_2 & b_3\\\cdashline{2-3}
          & a_3 & b_4
\end{array}
&
\begin{array}{c|*{4}{c}}
{S} & \underline{y} & \underline{z} & {d} & {r} \bigstrut\\\cline{2-5}
          & b_1 &  c_1 & d & 1\\
          & b_1 &  c_1 & d & 2\\\cdashline{2-5}
          & b_1 &  c_2 & d & 3\\\cdashline{2-5}
          & b_2 &  c_3 & d & 5\\
          & b_2 &  c_3 & d & 6\\\cdashline{2-5}
          & b_3 &  c_4 & d & 5\\\cdashline{2-5}
          & b_4 &  c_5 & d & 7\\
          & b_4 &  c_5 & e & 8       
\end{array}
&
\begin{array}{c|*{4}{c}}
M_{0}\defeq\interpret{\phi_{0}}{\db_{0}}  & x & y & z & r\bigstrut\\\cline{2-5}
          & a_1 & b_1 &  c_1 & 1\\
          & a_1 & b_1 &  c_1 & 2\\
          & a_1 & b_1 &  c_2 & 3\\
          & a_1 & b_2 &  c_3 & 5\\
          & a_1 & b_2 &  c_3 & 6\\
          & a_2 & b_2 &  c_3 & 5\\
          & a_2 & b_2 &  c_3 & 6\\ 
          & a_2 & b_3 &  c_4 & 5  
\end{array}
\end{array}
\]
\end{small}
\caption{Example database instance $\db_0$,
and the set~$M_{0}$ of all $\forall$embeddings of $q_{0}$ into $\db_{0}$..
}
\label{fig:exampleRewriting}
\end{figure}

\begin{figure}\centering
\begin{small}
\[
\begin{array}{ll}
\begin{array}{c|*{4}{c}c}
& \multicolumn{4}{l}{\fd{xyz}{r}}\\
\interpret{\phi_{1}}{M_{0}} & \underline{x} & \underline{y} & \underline{z} & r \bigstrut\\\cline{2-5}
          & a_1 & b_1 &  c_1 & 1\\
          & a_1 & b_1 &  c_1 & 2 & \bluekill\\\cdashline{2-5}
          & a_1 & b_1 &  c_2 & 3\\\cdashline{2-5}
          & a_1 & b_2 &  c_3 & 5\\
          & a_1 & b_2 &  c_3 & 6 & \bluekill\\\cdashline{2-5}
          & a_2 & b_2 &  c_3 & 5\\
          & a_2 & b_2 &  c_3 & 6 & \bluekill\\\cdashline{2-5}
          & a_2 & b_3 &  c_4 & 5  
\end{array}
&
\begin{array}{c|*{4}{c}c}
& \multicolumn{4}{l}{\fd{x}{y}}\\
\interpret{\phi_{2}}{M_{0}} & \underline{x} & y & z & r \bigstrut\\\cline{2-5}
          & a_1 & b_1 &  c_1 & 1\\
          & a_1 & b_1 &  c_2 & 3\\
          & a_1 & b_2 &  c_3 & 5 & \bluekill\\\cdashline{2-5}
          & a_2 & b_2 &  c_3 & 5\\
          & a_2 & b_3 &  c_4 & 5 & \redkill  
\end{array}
\end{array}
\]
\end{small}
\caption{Computation of a $\agsum$-minimal MCS relative to $\{\fd{x}{y}, \fd{yz}{r}\}$.}
\label{fig:simple}
\end{figure}

\begin{figure*}\centering
\begin{align*}
\sigma(x,y,z,r) & \defeq R(\underline{x},y)\land S(\underline{y,z},d,r)\land\forall v\forall r'\lrformula{S(\underline{y,z},v,r')\rightarrow v=d}\\
\phi_{0}(x,y,z,r) & \defeq \sigma(x,y,z,r)\land\forall y'\lrformula{R(\underline{x},y')\rightarrow\exists z'\exists r'\lrformula{\sigma(x,y',z',r')}}\\
    \phi_1(x,y,z,r) & \defeq \phi_{0}(x,y,z,r)\land 
    \forall r'\lrformula{\phi_{0}(x,y,z,r')\rightarrow r\leq r'}\\
    t(x,y) & \defeq\agterm{\agsum}{(z,r)}{r}{\phi_{1}(x,y,z,r)}\\
    \phi_2(x,y,z,r) & \defeq \phi_1(x,y,z,r)\land 
    \forall y'\forall z'\forall r'\lrformula{\phi_{1}(x,y',z',r')\rightarrow
    \lrformula{
    \begin{array}{l}
    t(x,y)\leq t(x,y')\land\\
    \lrformula{y'\prec y\rightarrow t(x,y)<t(x,y') }
    \end{array}}
    }
\\
\\
    \psi_2(x,v) & \defeq\exists y\exists z\exists r\lrformula{
    \begin{array}{l}
    \phi_1(x,y,z,r)\land\lrformula{v=t(x,y)}\land\\
    \forall y'\forall z'\forall r'\lrformula{\phi_{1}(x,y',z',r')\rightarrow
    \lrformula{
    t(x,y)\leq t(x,y')}
    }\end{array}
    }\\
    \glbcqa{g_{0}()} & \defeq \agterm{\agsum}{(x,v)}{v}{\psi_{2}(x,v)}
\end{align*}
\caption{
Calculation of both an $\agsum$-minimal MCS (formula $\phi_{2}$) and $\glbcqa{g_{0}()}$ for $\SUM(r)\leftarrow R(\underline{x},y), S(\underline{y,z},d,r)$.
}\label{fig:calcmcs}
\end{figure*}

Our example uses the database instance~$\db_{0}$ of Fig.~\ref{fig:exampleRewriting} and the following $\SUM$-query $g_{0}()$, in which $d$ is a constant:
\begin{equation}\label{eq:rewritex}\tag{$g_{0}()$}
\SUM(r)\leftarrow \overbrace{R(\underline{x},y), S(\underline{y,z},d,r)}^{q_{0}(x,y,z,r)}.
\end{equation}
The attack graph of the underlying Boolean conjunctive query has a single attack, from the $R$-atom to the $S$-atom.
Figure~\ref{fig:exampleRewriting} shows the set~$M_{0}$ of all $\forall$embeddings of $q_{0}$ in $\db_{0}$,
which can be calculated by the $\FOL$ formula $\phi_{0}(x,y,z,r)$ in Fig.~\ref{fig:calcmcs}.
Note incidentally that the embedding (in $\db_{0}$) which maps $(x,y,z,r)$ to $(a_{3},b_{4},c_{5},d,7)$ is not a $\forall$embedding, because of the value~$e$ ($e\neq d$) in the last row of the $S$-relation.
We have $\fdset{q_{0}}=\{\fd{x}{y}, \fd{yz}{r}\}$, and from Fig.~\ref{fig:exampleRewriting}, it is clear that $M_{0}\not\models\fdset{q_{0}}$.
We now introduce a notation for $\subseteq$-maximal subsets of $M_{0}$ that satisfy~$\fdset{q_{0}}$, and then provide a crucial lemma about such subsets.

\begin{definition}
    Let $q$ be a query in $\sjfbcq$, and $\db$ be a database instance.
    Let $M$ be a set of embeddings of $q$ in $\db$.
    A \emph{maximal consistent subset (MCS) of $M$} is a $\subseteq$-maximal subset $N$ of $M$ such that $N \models\fdset{q}$.
    We write $\mymcs{M}{q}$ for the subset of $\powerset{M}$ that contains all (and only) MCSs of~$M$.
\myqed
\end{definition}

\begin{lemma}\label{lem:mcs}
    Let $q$ be a query in $\sjfbcq$ with an acyclic attack graph. Let $\db$ be a database instance, and let $M$ be the set of all $\forall$embeddings of~$q$ in $\db$.
    Then,
    \begin{enumerate}
    \item \label{it:mcs1} for every superfrugal repair $\rep$ of $\db$, the set of all embeddings of~$q$ in~$\rep$ is an MCS of~$M$; and
    \item \label{it:mcs2} whenever $N$ is an MCS of~$M$, there is a superfrugal repair $\rep$ of $\db$ such that the set of embeddings of~$q$ in~$\rep$ is exactly~$N$.
\end{enumerate}
\end{lemma}

The following Corollary~\ref{cor:mcs}, which expresses $\glbcqa{g()}$ in terms of the construct of MCS, will be very helpful.
It requires monotonicity, but not associativity, of aggregate operators.


\begin{corollary}\label{cor:mcs}
    Let $g()\defeq\AGG(r)\leftarrow q(\vec{u})$ be a query in $\withaggr{\sjfbcq}$ such that the attack graph of $\exists\vec{u}\formula{q(\vec{u})}$ is acyclic. Let $\db$ be a database instance, and let $M$ be the set of all $\forall$embeddings of~$q$ in $\db$.
    If $\agagg$ is monotone, then
    \begin{equation}\label{eq:cormcs}
    \interpret{\glbcqa{g()}}{\db}=\min_{N\in\mymcs{M}{q}}\agagg(\bag{\theta(r)\mid \theta\in N}).
    \end{equation}
\end{corollary}
\begin{proof}
    Immediate corollary of Lemma~\ref{lem:mcs}.
\end{proof}

The proof of Theorem~\ref{the:glbmain} constructs an $\withaggr{\FOL}$-formula which correctly calculates the right-hand side of equation~\eqref{eq:cormcs}, provided that~$\agagg$ is not only monotone, but also associative.
In the remainder, we illustrate the construction by our running example, which uses $q_{0}$ and $M_{0}$. 

First, we focus on calculating an MCS $N^ {*}\in\mymcs{M_{0}}{q_{0}}$ at which the minimum of~\eqref{eq:cormcs} is reached, that is,
\[\agsum(\bag{\theta(r)\mid \theta\in N^ {*}})=\min_{N\in\mymcs{M_{0}}{q_{0}}}\agsum(\bag{\theta(r)\mid \theta\in N}).\] 
Informally, to obtain such an MCS~$N^{*}$, we must delete a $\subseteq$-minimal set of tuples from~$M_{0}$ in order to satisfy~$\{\fd{x}{y}, \fd{yz}{r}\}$, in a way that minimizes the $\SUM$ over the remaining $r$-values.
Such deletions are represented in Fig.~\ref{fig:calcmcs} by struck-through tuples.

 The left-hand table of Fig.~\ref{fig:simple} shows how, in a first step, we delete tuples from~$M_{0}$ in order to satisfy~$\fd{xyz}{r}$ (which is logically implied by~$\fd{yz}{r}$), in a way that minimizes the $\SUM$ over the remaining $r$-values: within each set of tuples that agree on $xyz$, we pick the one with the smallest $r$-value. The rationale for this step relies on the monotonicity of $\SUM$: smaller arguments will result in a smaller $\SUM$. 
 This computation is readily expressed by the formula $\phi_{1}(x,y,z,r)$ in Fig.~\ref{fig:calcmcs}, where for readability we assume a vocabulary with~$\leq$.
 In $\withaggr{\FOL}$, $t_{1}(\vec{x})\leq t_{2}(\vec{x})$ can be expressed as
 $t_{1}(\vec{x})=\agterm{\agmin}{v}{v}{v=t_{1}(\vec{x})\lor v=t_{2}(\vec{x})}$.
In this simple example, it is evident that the remaining tuples will satisfy $\fd{yz}{r}$, as desired, since the minimal value of~$r$ within each $xyz$-group is completely determined by $yz$, regardless of~$x$. Proving that such independences hold in general is a major challenge in the general proof of Theorem~\ref{the:glbmain}.
 
 The next step is to delete more tuples in order to also satisfy~$\fd{x}{y}$, which is illustrated by the right-hand table of Fig.~\ref{fig:simple}.
 Within each group of tuples that agree on~$x$, we pick the $y$-value that results in the smallest $\SUM$.
 In the example, since $1+3<5$, the tuple with $r$-value~$5$ is deleted, indicated by a blue strike-through.
 There is a tie among the tuples where the $x$-value equals~$a_{2}$. To break this tie, we opt for the smallest $y$-value according to lexicographical order. Specifically, $b_{2}$ is chosen over $b_{3}$, indicated by a red strike-through.
 This computation is expressed by the formula $\phi_{2}(x,y,z,r)$ in Fig.~\ref{fig:calcmcs}, where $\preceq$ is used for lexicographical order.
 
 To conclude this example, we note that to calculate the value at the right-hand side of~\eqref{eq:cormcs}, there is actually no need to entirely compute the right-hand table of Fig.~\ref{fig:simple}. We only need to know, for every set of tuples that share the same $x$-value, the total sum of their $r$-values. This is achieved by the formula $\psi_{2}$ in Fig.~\ref{fig:calcmcs}, which does not rely on lexicographical order. In our example, $\psi_{2}(a_{1},4)$ ($4=3+1$) and $\psi_{2}(a_{2},5)$ hold true.
 Finally, $\glbcqa{g_{0}()}$ is obtained by applying $\SUM$ over the second column of $\psi_{2}$, which yields~$9$ in our example. 
 Note that associativity is needed to enable incremental aggregation.




\subsection{Aggregation Queries with Free Variables}\label{sec:freevariables}

So far, we have focused on numerical terms $g()$ without free variables.
We now explain how our results extend to numerical terms $g(\vec{x})=\agterm{\calF}{\vec{y}}{r}{q(\vec{x},\vec{y})}$ with free variables~$\vec{x}\defeq(x_{1},\ldots,x_{k})$, where $q(\vec{x},\vec{y})$ is self-join-free.
Let $\vec{c}=(c_{1},\ldots,c_{k})$ be a sequence of distinct constants.
Let $q_{\vec{c}}(\vec{y})$ be the conjunction obtained from $q(\vec{x},\vec{y})$ by replacing, for $i\in\{1,\ldots,k\}$, each occurrence of each~$x_{i}$ by~$c_{i}$. 
Then,
\begin{itemize}
    \item 
if the attack graph of $\exists\vec{y}\formula{q_{\vec{c}}(\vec{y})}$ is cyclic, then by Theorem~\ref{the:inexpressible}, $\glbcqa{g(\vec{c})}$ is not in $\withaggr{\FOL}$; and 
    \item 
if the attack graph of $\exists\vec{y}\formula{q_{\vec{c}}(\vec{y})}$ is acyclic, then by Theorem~\ref{the:expressible}, there is a formula $\varphi_{\vec{c}}$ in $\withaggr{\FOL}$ that solves $\glbcqa{g(\vec{c})}$.     
\end{itemize}
It is now not hard to show that since $q(\vec{x},\vec{y})$ is self-join-free, different sequences of constants involve the same calculations up to a renaming of constants.
This implies that we can perform the calculation once by treating the free variables in~$\vec{x}$ as distinct constants. This treatment of free variables is often used in consistent query answering (and in logic in general~\cite[Lemma~2.3]{DBLP:books/sp/Libkin04}); however, it fails for CQA in the presence of self-joins.
This is one of the reasons why assuming self-join-freeness serves as a simplifying assumption in much work on CQA.


\section{Aggregate Operators Lacking  Monotonicity or Associativity}\label{sec:lacking}

Let $g()\defeq\AGG(r)\leftarrow q(\vec{y})$ be a numerical query in~$\withaggr{\sjfbcq}$.
In this section, we explore the scenario left open by Theorems~\ref{the:inexpressible} and~\ref{the:expressible}: we assume that the attack graph of $\exists\vec{y}\formula{q(\vec{y})}$ is acyclic, but that $\agagg$ lacks either monotonicity, associativity, or both.
We introduce in Section~\ref{sec:dc} a manifestation of non-monotonicity, called \emph{descending chains}, which is used to identify cases where $\glbcqa{g()}$ is not expressible in $\withaggr{\FOL}$ under the above assumptions.
Descending chains are also used in Section~\ref{sec:ub} in the study of $\lubcqa{g()}$ via the use of dual aggregate operators, and in Section~\ref{sec:unconstrained} to demonstrate that unconstrained numeric columns (i.e., columns not constrained to $\rationals_{\geq 0}$) suffice for moving from expressibility in $\withaggr{\FOL}$ to non-expressibility.
We conclude with a positive result: when $g()$ is a  $\MIN$- or a $\MAX$-query, it is decidable whether or not $\glbcqa{g()}$ and $\lubcqa{g()}$ are expressible in $\withaggr{\FOL}$.



\subsection{Aggregate Operators with Descending Chains}\label{sec:dc}

We show that the inverse of Theorem~\ref{the:inexpressible} does not hold: Lemmas~\ref{lem:dc} and~\ref{lem:bdc} introduce numerical queries with acyclic attack graphs that however do not allow glb rewriting in $\withaggr{\FOL}$.
Their proofs rely on the existence of a (possibly bounded) descending chain for an aggregate operator~$\agagg$, which implies that $\agagg$ lacks monotonicity. 
Concrete examples of such aggregate operators are $\agavg$ and $\agprod$, as expressed by Corollary~\ref{cor:avgprod} (assuming that the numeric domain is $\rationals_{\geq 0}$).

\begin{definition}[Descending chain]\label{def:dc}
For $i\in\naturals$ and $t\in\rationals$, we write $\copies{i}{t}$ as a shorthand for $\overbrace{t,t,\ldots,t}^{\textnormal{$i$ times}}$, i.e., $i$ occurrences of~$t$.
We say that an aggregate operator $\agagg$ has a \emph{descending chain} if there exist $s,t\in\rationals_{\geq 0}$ such that for every $i\in\naturals$,
$\agagg(\bag{s,\copies{i}{t}})>\agagg(\bag{s,\copies{(i+1)}{t}})$.
Note that $s$ and $t$ need not be distinct.
Such a descending chain is said to be \emph{bounded} if for every $i\in\naturals$,
there exists $m_{i}\in\rationals_{\geq 0}$ such that for all $j\in\naturals_{>0}$, for all $k',k\in\naturals$ such that $k'\leq k\leq i$,  we have $\agagg(\bag{s,\copies{k'}{t}})<\agagg(\bag{\copies{j}{m_{i}},s,\copies{k}{t}})$. 
Informally, this expression means that $\agagg$ will strictly increase if at least one copy of~$m_{i}$ is added, regardless of any addition of copies of~$t$.
Notice that $m_{i}$ depends on~$i$.
See the proof of Lemma~\ref{lem:chainavgmod} for examples.
\myqed
\end{definition}


\begin{lemma}\label{lem:dc}
Let $\agagg$ be an aggregate operator with a descending chain (which may not be bounded).
Then, $\glbcqa{g()}$ is $\NL$-hard for 
$g()\defeq\AGG(r)\leftarrow R(\underline{x},y,r), S_{1}(\underline{y},x), S_{2}(\underline{y},x)$.
Consequently, $\glbcqa{g()}$ is not expressible in $\withaggr{\FOL}$.
\end{lemma}

\begin{lemma}\label{lem:bdc}
Let $\agagg$ be an aggregate operator with a bounded descending chain.
Then, $\glbcqa{g()}$ is $\NP$-hard for $g()\defeq
\AGG(r)\leftarrow S_{1}(\underline{x},c_{1}), S_{2}(\underline{y},c_{2}), T(\underline{x,y},r)$, where $c_{1}$ and $c_{2}$ are (not necessarily distinct) constants.
Consequently, $\glbcqa{g()}$ is not expressible in $\withaggr{\FOL}$.
\end{lemma}

\begin{lemma}\label{lem:chainavgmod}
    $\agavg$ and $\agprod$ have bounded descending chains.
\end{lemma}
\begin{proof}
    For $\agavg$, we have $\agavg(\bag{1})>\agavg(\bag{1,0})>\agavg(\bag{1,0,0}>\agavg(\bag{1,0,0,0}>\dotsm$.
    Take $s=1$ and $t=0$.
    For every $i\in\naturals$, let $m_{i}=i+2$.
    Whenever $j>0$ and $0\leq k'\leq k\leq i$, we have
    $\frac{1}{k'+1}=\agavg(\bag{1,\copies{k'}{0}}<\agavg(\bag{\copies{j}{(i+2)},1,\copies{k}{0}}=\frac{j*(i+2)+1}{j+k+1}$, as desired.

    For $\agprod$, we have $\agprod(\bag{\frac{1}{2}})>\agprod(\bag{\frac{1}{2},\frac{1}{2}})>\agprod(\bag{\frac{1}{2},\frac{1}{2},\frac{1}{2}})>\dotsm$.
    Take $s=t=\frac{1}{2}$.
    Let $i\in\naturals$.
    For $j>0$ and $0\leq k'\leq k\leq i$, we obtain
    $\agprod(\bag{\frac{1}{2},\copies{k'}{\frac{1}{2}}}<\agprod(\bag{\copies{j}{m_{i}},\frac{1}{2},\copies{k}{\frac{1}{2}}}$ by choosing $m_{i}=2^{i+1}$.
\end{proof}

\begin{corollary}\label{cor:avgprod}
    $\glbcqa{g()}$ is not expressible in $\withaggr{\FOL}$ with $g()$ as in Lemma~\ref{lem:dc} or Lemma~\ref{lem:bdc} and $\AGG\in\{\AVG, \PROD\}$.
\end{corollary}



\subsection{Dual Aggregate Operators}\label{sec:ub}

\begin{definition}[Dual aggregate operator]
The \emph{dual} of a positive aggregate operator~$\calF$, denoted $\dual{\calF}$, is defined as the function that takes as argument a finite multiset~$X$ of non-negative rational numbers such that $\dual{\calF}(X)=-1*\calF(X)$ if $X\neq\emptyset$, and $\dual{\calF}(\emptyset)=\calF(\emptyset)$. 
The dual of a positive aggregate operator is also called a \emph{dual aggregate operator}.
\myqed
\end{definition}

Note that a dual aggregate operator is not a positive aggregate operator itself, as it can return negative rational numbers. Despite this, Definition~\ref{def:dc} of (bounded) descending chain applies also to dual aggregate operators. Additionally, it can be verified that Lemmas~\ref{lem:dc} and~\ref{lem:bdc} remain valid for dual aggregate operators, since their proofs do not rely on the signs of the aggregated values.

As argued in Section~\ref{sec:introduction} and formalized by Proposition~\ref{pro:dual}, the function problem $\lubcqa{g()}$ for $g()\defeq\AGG(r)\leftarrow q(\vec{u})$ is the same up to a sign as $\glbcqa{h()}$ where $h()$ has the same body as~$g()$, but uses the dual of $\AGG$ in its head:
\[
h()\defeq\dual{\AGG}(r)\leftarrow q(\vec{u}),
 \]
where the aggregate symbol $\dual{\AGG}$ is interpreted by $\dual{\agagg}$, i.e., the dual of $\agagg$.
 The semantics of~$h()$ on a database instance $\db$ is naturally defined: if $g()$ returns a rational number $r$ distinct from~$f_{0}$, then $h()$ returns $-1*r$; and if $g()$ returns~$f_{0}$, then so does~$h()$. 
 The problems $\glbcqa{h()}$ and $\lubcqa{h()}$ are defined as before.
 The proof of the following propostion is straightforward.

 \begin{proposition}\label{pro:dual}
Let $g()\defeq\AGG(r)\leftarrow q(\vec{u})$ be a numerical query in $\withaggr{\sjfbcq}$. 
Let $h()\defeq\dual{\AGG}(r)\leftarrow q(\vec{u})$.   
Then, for every database instance $\db$ such that $\db\cqamodels\exists\vec{u}\formula{q(\vec{u})}$, we have
$\interpret{\lubcqa{g()}}{\db}=-1*\interpret{\glbcqa{h()}}{\db}$.
\end{proposition}

If we let $\AGG=\SUM$ in the numerical term $g()$ of Lemma~\ref{lem:dc}, then $\glbcqa{g()}$ is in $\withaggr{\FOL}$ (by Theorem~\ref{the:expressible}), but $\lubcqa{g()}$ is not, as a consequence of the following lemma.


\begin{theorem}\label{the:lubchain}
    $\lubcqa{g()}$ is not expressible in $\withaggr{\FOL}$ with $g()$ as in Lemma~\ref{lem:dc} and $\AGG\in\{\SUM, \AVG, \PROD\}$.
\end{theorem}
\begin{proof}
    From $\dual{\agavg}(\bag{0})>\dual{\agavg}(\bag{0,1})>\dual{\agavg}(\bag{0,1,1}>\dual{\agavg}(\bag{0,1,1,1}>\dotsm$, it follows that $\dual{\agavg}$ has a descending chain.
    From $\dual{\agsum}(\bag{1})>\dual{\agsum}(\bag{1,1})>\dual{\agsum}(\bag{1,1,1}>\dotsm$, it follows that $\dual{\agsum}$ has a descending chain.
    From $\dual{\agprod}(\bag{2})>\dual{\agprod}(\bag{2,2})>\dual{\agprod}(\bag{2,2,2}>\dotsm$, it follows that $\dual{\agprod}$ has a descending chain.
    The desired result then follows by Proposition~\ref{pro:dual} and~Lemma~\ref{lem:dc}.
\end{proof}

It can be easily verified that the descending chain for $\PROD$ in the proof of Theorem~\ref{the:lubchain} is bounded by choosing $m_{i}=\frac{1}{2^{i+1}}$, which implies that $\lubcqa{g()}$  is not expressible in $\withaggr{\FOL}$ with $g()$ as in Lemma~\ref{lem:bdc} and $\AGG=\PROD$.

\subsection{Unconstrained Numerical Columns}\label{sec:unconstrained}

So far, we have restricted our attention to database instances in which all numbers occurring in numeric columns are non-negative.
For the following theorem, it is relevant to note that while $\agsum$ is monotone over both $\naturals$ and $\rationals_{\geq 0}$ (and covered by Theorem~\ref{the:glbmain}), it becomes non-monotone if these domains are extended by even a single negative number (in our treatment, the integer $-1$).
Since the numerical query $g()$ of Theorem~\ref{the:dixit} is in $\caggforest$, it disproves a claim in~\cite{FuxmanThesis} stating that $\glbcqa{g()}$ is expressible in (some restricted) aggregate logic for all numerical queries in $\caggforest$.
The definition of $\caggforest$ is in Appendix~\ref{sec:caggforest}.


\begin{theorem}\label{the:dixit}
Assume that the third attribute of~$T$ is a numeric column that can contain numbers in $\naturals\cup\{-1\}$.
Then,
$\glbcqa{g()}$ is $\NP$-hard for $g()\defeq
\SUM(r)\leftarrow S_{1}(\underline{x},c_{1}), S_{2}(\underline{y},c_{2}), T(\underline{x,y},r)$, where~$c_{1}$ and~$c_{2}$ are (not necessarily distinct) constants.
Consequently, $\glbcqa{g()}$ is not in $\withaggr{\FOL}$.
\end{theorem}
\begin{proof}[Proof sketch.]
    It can be easily verified that $\agsum$ has a bounded descending chain if $-1$ can be used.
    The desired result then follows from Lemma~\ref{lem:bdc}.
\end{proof}

\subsection{$\MIN$ and $\MAX$}

$\agmin$ is not monotone, since its value can decrease when extending a multiset, for example,
$\agmin(\bag{3})>\agmin(\bag{2,3})$.
Despite this, Theorem~\ref{the:expressible} extends to~$\MIN$:

\begin{theorem}\label{the:minnew}
Let $g()\defeq\MIN(r)\leftarrow q(\vec{y})$ be a query in $\withaggr{\sjfbcq}$ such that the attack graph of $\exists\vec{y}\formula{q(\vec{y})}$ is acyclic.
Then, $\glbcqa{g()}$ is expressible in $\withaggr{\FOL}$.
\end{theorem}

The following theorem states that rewritability in $\withaggr{\FOL}$ is decidable for $\MIN$- and $\MAX$-queries, for both glb and lub.

\begin{theorem}[Separation Theorem for $\MIN$ and $\MAX$]\label{the:sepminmax}
Let $g()\defeq\AGG(r)\leftarrow q(\vec{y})$ be a numerical query in $\withaggr{\sjfbcq}$ with $\AGG\in\{\MIN, \MAX\}$.
If the attack graph of $\exists\vec{y}\formula{q(\vec{y})}$ is acyclic,
then both $\lubcqa{g()}$ and $\glbcqa{g()}$ are expressible in $\withaggr{\FOL}$;
otherwise neither $\lubcqa{g()}$ nor $\glbcqa{g()}$ is expressible in $\withaggr{\FOL}$.
\end{theorem}


\section{Summary and Open Questions}\label{sec:discussion}

We studied the complexity of computing range consistent answers to numerical queries $g()$ of the form $\AGG(r)\leftarrow q(\vec{u})$ where $q(\vec{u})$ is a self-join-free conjunction of atoms.
Given a database instance that may violate its primary key constraints, the question is to determine the minimal (glb) and maximal (lub) values of $g()$ returned across all repairs.
Since the lub coincides (up to a sign) with the glb relative to the dual aggregate operator $\dual{\AGG}$, it suffices to focus on the glb, i.e., on $\glbcqa{g()}$. 
Our main result is that for aggregate operators $\AGG$ that are monotone and associative (e.g., $\SUM$ over $\rationals_{\geq 0}$), it is decidable, given $g()$, whether or not $\glbcqa{g()}$ can be expressed in the aggregate logic~$\withaggr{\FOL}$.
Decidability also holds for $\AGG\in\{\MIN, \MAX\}$, for both glb and lub.
On top of these complete dichotomies, we obtained some results of inexpressibility in $\withaggr{\FOL}$ for aggregate operators that lack monotonicity or associativity (e.g., $\AVG$).
We also refuted a longstanding claim made in~\cite{FuxmanThesis}.

An open question remains as follows. Let $\calC$ denote the class of $\AGG$-queries $g()\defeq\AGG(r)\leftarrow q(\vec{y})$, without self-joins, that satisfy all of the following: \emph{(i)}~the attack graph of $\exists y\formula{q(\vec{y})}$ is acyclic, \emph{(ii)}~$\AGG$ lacks monotonicity, associativity, or both, and \emph{(iii)}~$\AGG$ is neither $\MIN$ nor $\MAX$. That is, $\calC$ contains those numerical queries not covered by our results. It remains open to determine, for \emph{each} query $g()$ in~$\calC$, whether or not $\glbcqa{g()}$ is expressible in~$\withaggr{\FOL}$.
We only showed inexpressibility for specific queries  in $\calC$ (cf.\ Corollary~\ref{cor:avgprod}).

\revision{
Another open question concerns shifting our focus from expressibility in $\withaggr{\FOL}$ to computability in $\PTIME$.
By changing the focus of Theorem~\ref{the:glbmain}, we can formulate the following conjecture:
\begin{conjecture}\label{con:glbmain}
    Given a numerical query $g()$ in $\withaggr{\sjfbcq}$ whose aggregate operator is both monotone and associative, $\glbcqa{g()}$ is either in $\PTIME$ or $\coNP$-hard, and it can be decided which of the two cases holds.
\end{conjecture}
We now outline a possible route to proving Conjecture~\ref{con:glbmain}.
In~\cite{DBLP:journals/tods/KoutrisW17}, each cycle in the attack graph of a query~$q\in\sjfbcq$ is classified as either \emph{weak} or \emph{strong}, which is a decidable property. 
It is then shown that the decision problem $\cqa{q}$ is in~$\PTIME$ if $q$'s attack graph contains no strong cycles, and $\cqa{q}$ is $\coNP$-complete otherwise. 
Let  $g()\defeq\AGG(r)\leftarrow q(\vec{u})$ be a numerical query in $\withaggr{\sjfbcq}$.
Since a solution to $\glbcqa{g()}$ also solves $\cqa{\exists\vec{u}\formula{q(\vec{u})}}$, it follows that $\glbcqa{g()}$ is $\coNP$-hard if the attack graph of $\exists\vec{u}\formula{q(\vec{u})}$ contains a strong cycle.
To show Conjecture~\ref{con:glbmain}, it suffices therefore to establish that if all cycles in the attack graph of $\exists\vec{u}\formula{q(\vec{u})}$ are weak, and $\agagg$ is both monotone and associative, then $\glbcqa{g()}$ is in~$\PTIME$. For attack graphs without cycles, the latter follows from Theorem~\ref{the:glbmain} under the assumption that $\withaggr{\FOL}$ is in $\PTIME$ (which requires that aggregate operators are computable in polynomial time).
So the remaining problem concerns handling weak cycles in attack graphs, which we anticipate might be solvable by employing the constructs developed in~\cite[Section~8]{DBLP:journals/mst/KoutrisW21} or~\cite[Section~4]{DBLP:conf/icdt/FigueiraPSS23}. Specificallly, it should be investigated whether $\forall$embeddings can be generalized to address weak cycles, enabling $\glbcqa{g()}$ to be solved along the lines presented in Section~\ref{sec:lb} of the current paper---i.e., by computing the smallest aggregated value over all $\subseteq$-maximal consistent sets of $\forall$embeddings.  
}




\appendix

\section{Helping Lemmas}

We introduce two helping lemmas that will be used later on.

\begin{lemma}[\mbox{\cite[Lemma A.3]{DBLP:conf/icdt/KhalfiouiW23}}] \label{lem:attackSameAtom}
    Let $q$ be a query in $\sjfbcq$. 
    Let $R$ and $S$ be two distinct atoms that both attack a same atom. If $R \nattacks{q} S$, then either $S \attacks{q} R$ or $\key{S} \subseteq \keycl{R}{q}$.
\end{lemma}

\begin{definition}[Sequential proof]
    Let $q$ be a query in $\sjfbcq$.
    Let $Z\subseteq\vars{q}$ and $w\in\vars{q}$.
A \emph{sequential proof} of $\fdset{q}\models\fd{Z}{w}$ is a (possibly empty) sequence
$\tuple{F_{1}, F_{2}, \ldots, F_{n}}$ of atoms in $q$ such that for every $i\in \{1,\ldots,n\}$, 
$\key{F_i}\subseteq Z\cup\lrformula{\bigcup_{j=1}^{i-1} \vars{F_j}}$ and $w\in Z\cup\lrformula{\bigcup_{j=1}^{n} \vars{F_j}}$.
\myqed
\end{definition}

\begin{lemma}\label{lem:backwardinseqproof}
    Let $q$ be query in $\sjfbcq$.
    For some~$w\in\vars{q}$, let $(F_{1},F_{2},\ldots,F_{n})$ be a sequential proof of $\fdset{q\setminus\{G\}}\models\fd{Z}{w}$.
    If $G\attacks{q}x$ for some variable $x$ in $Z\cup\lrformula{\bigcup_{i=1}^{n}\vars{F_{i}}}$,
    then there is $z\in Z$ such that $G\attacks{q}z$. 
\end{lemma}
\begin{proof}
Assume that $G\attacks{q}x$ for some variable $x$ in $Z\cup\lrformula{\bigcup_{i=1}^{n}\vars{F_{i}}}$.
The proof runs by induction on increasing $n$.    
For the basis of the induction, $n=0$, we obtain $x\in Z$ with $G\attacks{q}x$, as desired. 

For the induction step, $n-1\rightarrow n$, assume that the lemma holds for sequential proofs of length~$<n$.
Notice that for every $m\in\{0,1,\ldots,n-1\}$, there is some $w'$ such that  $(F_{1},F_{2},\ldots,F_{m})$ is a sequential proof of $\fdset{q\setminus\{G\}}\models\fd{Z}{w'}$.
Therefore, if $x$ occurs in $Z\cup\lrformula{\bigcup_{i=1}^{n-1}\vars{F_{i}}}$, then the desired result holds by the induction hypothesis.
Assume from here on that $x\notin Z\cup\lrformula{\bigcup_{i=1}^{n-1}\vars{F_{i}}}$.
Then, $x\in\notkey{F_{n}}$.
From $G\attacks{q}x$, it follows that there is $x'\in\key{F_{n}}$ such that $G\attacks{q}x'$.
By the definition of sequential proof, it must be the case that $x'\in Z\cup\lrformula{\bigcup_{i=1}^{n-1}\vars{F_{i}}}$.
Again, the desired result holds by the induction hypothesis.
\end{proof}

\begin{lemma} \label{lem:attacksSameAfterValuation}
    Let $q$ be a query in $\sjfbcq$.
    Let $\theta$ be a valuation over a subset of $\vars{q}$.
    Let $F \in q$ such that for every variable $v \in \domain{\theta}$, $F \nattacks{q} v$.
    Then, for every variable $v \in \vars{q} \setminus \domain{\theta}$,
    $F \attacks{q} v$ if and only if $\theta(F) \attacks{\theta(q)} v$.
\end{lemma}
\begin{proof}
    \framebox{$\implies$} 
    Let $v \in \vars{q} \setminus \domain{\theta}$ such that
    $F \attacks{q} v$.
    Since $F \attacks{q} v$, there is a sequence of variables $\tuple{w_1, \ldots, w_n}$ such that $w_1 \in \notkey{F}$, $w_n = v$, every two adjacent variables appear together in some atom of $q$, and for each $j \in \{1, \ldots, n\}$, we have $F \attacks{q} w_j$, which implies $w_j \not \in \keycl{F}{q}$ and $w_j \not \in \domain{\theta}$.
    If for every $j \in \{1, \ldots, n\}$, $w_j \not \in \keycl{\theta(F)}{\theta(q)}$, then it follows that $\theta(F) \attacks{\theta(q)} v$.
    Assume, for the sake of contradiction, that there is a $j \in \{1, \ldots, n\}$ such that $w_j \in \keycl{\theta(F)}{\theta(q)}$.
    It follows that, $\fdset{\theta(q) \setminus \{\theta(F)\}} \models \fd{\key{\theta(F)}}{w_j}$.
Consequently, $\fdset{q \setminus \{F\}} \models \fd{\key{F} \cup \domain{\theta}}{w_j}$.
Since $F\attacks{q}w_{j}$, it follows by Lemma~\ref{lem:backwardinseqproof} that $F\attacks{q}w$ for some variable $w \in \key{F} \cup \domain{\theta}$, a contradiction.
    \framebox{$\impliedby$}
    Easy.
\end{proof}

\section{Proof of Lemma~\ref{lem:anysort}}

\begin{proof}[Proof of Lemma~\ref{lem:anysort}]
Assume that $\theta$ is an $n$-$\forall$embedding of~$q$ in $\db$ relative to $(q,\sort_{1})$.
We first show that the desired result holds for a topological sort $(q,\sort_{2})$ obtained by swapping two adjacent atoms, say~$F_{k}$ and $F_{k+1}$.  
To this end, let 
\[
\setlength{\arraycolsep}{1pt}
\begin{array}{lll}
    (q,\sort_{2}) & =(F_{1},\ldots,F_{k-1}, F_{k},F_{k+1}, F_{k+2},\ldots,F_{n}),\\
     &\phantom{F_{1},\ldots,F_{k-1},F_{k+1}}\nearrow\!\!\!\!\!\!\nwarrow\\
    (q,\sort_{1}) & =(F_{1},\ldots,F_{k-1},F_{k+1}, F_{k}, F_{k+2},\ldots,F_{n}).
\end{array}
\]
We have $F_{k}\nattacks{q}F_{k+1}$ and $F_{k+1}\nattacks{q}F_{k}$.

Let $\theta'$ be the restriction of $\theta$ to $\bigcup_{i=1}^{k-1}\vars{F_{i}}$.
Let $p=\{\theta'(F_{k})$, $\theta'(F_{k+1})$, \dots, $\theta'(F_{n})\}$, and for $i\in\{1,2,\ldots,n-k+1\}$, let $G_{i}=\theta'(F_{i+k-1})$.
Let 
\begin{align*}
    (p,\sort_{2}') &=\tuple{G_{1},G_{2},G_{3},\ldots,G_{n-k+1}},\\
    (p,\sort_{1}') &=\tuple{G_{2},G_{1},G_{3},\ldots,G_{n-k+1}},
\end{align*}
where the topological sorts $\sort_{2}'$ and $\sort_{1}'$ are inherited from $\sort_{2}$ and~$\sort_{1}$.
It is known that these are indeed topological sorts of $p$'s attack graph.
Informally, $p$ is obtained from $q$ by first applying the partial valuation $\theta'$ on~$q$ and then omitting the facts $\theta'(F_{1})$, $\theta'(F_{2})$, \dots, $\theta'(F_{k-1})$.
We have $G_{1}\nattacks{p}G_{2}$ and $G_{2}\nattacks{p}G_{1}$.
Let $\mu$ be the restriction of $\theta$ to $\vars{p}$.
It can be verified from our definitions that 
\begin{equation}\label{eq:muforallpath}
\mbox{$\mu$ is a $(n-k+1)$-$\forall$embedding of~$p$ in~$\db$ relative to  $(p,\sort_{1}')$}.
\end{equation}
It suffices now to show that $\mu$ is a $(n-k+1)$-$\forall$embedding of~$p$ in~$\db$ relative to  $(p,\sort_{2}')$.

For $i\in\{1,2\}$, let $G_{i}=R_{i}(\underline{\vec{x}_{i}},\vec{y}_{i})$, in which $\vec{x}_{i}\vec{y}_{i}$ needs not be constant-free, and let $\mu(G_{i})=R_{i}(\underline{\vec{a}_{i}},\vec{b}_{i})$.
By~\eqref{eq:muforallpath}, 
\begin{align}
\db & \cqamodels\substitute{p}{\vec{x}_{2}}{\vec{a}_{2}}\label{eq:inst}\\
\db & \cqamodels\substitute{\lrformula{p\setminus\{G_{2}\}}}{\vec{x}_{2}\vec{y}_{2}\vec{x}_{1}}{\vec{a}_{2}\vec{b}_{2}\vec{a}_{1}}\label{eq:instonetwo}
\end{align}

To show the desired result for one swap, it suffices to show $\db\cqamodels\substitute{p}{\vec{x}_{1}}{\vec{a}_{1}}$ and $\db\cqamodels\substitute{\lrformula{p\setminus\{G_{1}\}}}{\vec{x}_{1}\vec{y}_{1}\vec{x}_{2}}{\vec{a}_{1}\vec{b}_{1}\vec{a}_{2}}$.
In what follows, a fact is said to be \emph{relevant} for a conjunctive query in a database instance if there is an embedding that maps a query atom to the fact.
\begin{description}
\item[Proof that $\db\cqamodels\substitute{p}{\vec{x}_{1}}{\vec{a}_{1}}$.]
Let $\rep$ be an arbitrary repair of~$\db$.
Let $\vec{b}_{2}'$ be the (unique) sequence of constants, of length~$\card{\vec{y}_{2}}$, such that $R_{2}(\underline{\vec{a}_{2}},\vec{b}_{2}')\in\rep$.
From~\eqref{eq:inst}, it follows $\rep\models\substitute{p}{\vec{x}_{2}}{\vec{a}_{2}}$, hence $R_{2}(\underline{\vec{a}_{2}},\vec{b}_{2}')$ is relevant for~$p$ in~$\rep$.
From~\eqref{eq:instonetwo}, it follows
\begin{equation*}
    \rep\models\substitute{\lrformula{p\setminus\{G_{2}\}}}{\vec{x}_{2}\vec{y}_{2}\vec{x}_{1}}{\vec{a}_{2}\vec{b}_{2}\vec{a}_{1}}.
\end{equation*}
Consequently,
\begin{equation*}
    \lrformula{\rep\setminus\{R_{2}(\underline{\vec{a}_{2}},\vec{b}_{2}')\}}\cup\{R_{2}(\underline{\vec{a}_{2}},\vec{b}_{2})\}\models\substitute{p}{\vec{x}_{2}\vec{y}_{2}\vec{x}_{1}}{\vec{a}_{2}\vec{b}_{2}\vec{a}_{1}}.
\end{equation*}
It follows
\begin{equation*}
    \lrformula{\rep\setminus\{R_{2}(\underline{\vec{a}_{2}},\vec{b}_{2}')\}}\cup\{R_{2}(\underline{\vec{a}_{2}},\vec{b}_{2})\}\models\substitute{p}{\vec{x}_{1}}{\vec{a}_{1}}.
\end{equation*}
Since $G_{2}\nattacks{p}G_{1}$, it follows by~\cite[Lemma~B.1]{DBLP:journals/tods/KoutrisW17} that
\begin{equation}\label{eq:relevant}
    \rep\models\substitute{p}{\vec{x}_{1}}{\vec{a}_{1}}.
\end{equation}
\item[Proof that $\db\cqamodels\substitute{\lrformula{p\setminus\{G_{1}\}}}{\vec{x}_{1}\vec{y}_{1}\vec{x}_{2}}{\vec{a}_{1}\vec{b}_{1}\vec{a}_{2}}$.]
Let $\rep$ be an arbitrary repair of~$\db$.
Let $\vec{b}_{1}'$ be the (unique) sequence of constants, of length~$\card{\vec{y}_{1}}$, such that $R_{1}(\underline{\vec{a}_{1}},\vec{b}_{1}')\in\rep$.
From~\eqref{eq:relevant}, it follows that $R_{1}(\underline{\vec{a}_{1}},\vec{b}_{1}')$ is relevant for~$p$ in~$\rep$.
By~\eqref{eq:inst},
\begin{equation*}
    \rep\models\substitute{p}{\vec{x}_{2}}{\vec{a}_{2}}.
\end{equation*}
Since $G_{1}\nattacks{p}G_{2}$, it follows by~\cite[Lemma~B.1]{DBLP:journals/tods/KoutrisW17} that
\begin{equation*}
    \lrformula{\rep\setminus\{R_{1}(\underline{\vec{a}_{1}},\vec{b}_{1}')\}}\cup\{R_{1}(\underline{\vec{a}_{1}},\vec{b}_{1})\}\models\substitute{p}{\vec{x}_{2}}{\vec{a}_{2}}.
\end{equation*}
Consequently,
\begin{equation*}
    \rep\models\substitute{\lrformula{p\setminus\{G_{1}\}}}{\vec{x}_{1}\vec{y}_{1}\vec{x}_{2}}{\vec{a}_{1}\vec{b}_{1}\vec{a}_{2}}.
\end{equation*}
This concludes the proof for one swap.
\end{description}
To conclude the proof of Lemma~\ref{lem:anysort}, it suffices to observe that every topological sort can be obtained from $(q,\sort_{1})$ by zero, one, or more swaps, and that any swap results in an $n$-$\forall$embedding.
\end{proof}

\section{Proof of Lemma~\ref{lem:allfo}}\label{sec:frugalformula}

\begin{proof}[Proof of Lemma~\ref{lem:allfo}]
We define a formula $\frugalformula{q}(\vec{u})$ such that for every database instance $\db$,
$\db\models\frugalformula{q}(\vec{c})$ if and only if the valuation $\theta$ over $\vars{\vec{u}}$ such that $\theta(\vec{u})=\vec{c}$ is an $n$-$\forall$embedding in $\db$.

Let $p(\vec{x})$ be a conjunctive query with free variables~$\vec{x}$.
A \emph{consistent first-order rewriting of $p(\vec{x})$} is a first-order formula $\omega(\vec{x})$ such that for every database instance $\db$, for every sequence $\vec{c}$ of constants, of length $\card{\vec{x}}$, we have $\db\cqamodels p(\vec{c})$ if and only if $\db\models\omega(\vec{c})$.
The following problem has been solved in~\cite{DBLP:journals/tods/KoutrisW17}:
given a self-join-free conjunctive query $p(\vec{x})$, decide whether $p(\vec{x})$ has a first-order rewriting, and if affirmative, construct such a a first-order rewriting. 

Let $q$ be a query in $\sjfbcq$ with an acyclic attack graph.
Assume that $q$'s body is $F_{1}\land F_{2}\land\dotsm\land F_{n}$, where the atoms are listed in a topological sort of the attack graph.
For $j\in\{0, 1, 2, \ldots, n\}$, we inductively define a formula $\psi_{j}(\vec{u}_{j})$ expressing that $\vec{u}_{j}$ is a $j$-$\forall$embedding.
The basis of the induction is $\psi_{0}()=\mathsf{true}$.
For the induction step, $j\rightarrow j+1$, the formula $\psi_{j+1}(\vec{u}_{j+1})$ reads as follows:
\[
\psi_{j+1}(\overbrace{\vec{u}_{j},\vec{x}_{j+1},\vec{y}_{j+1}}^{\vec{u}_{j+1}})\defeq
\psi_{j}(\vec{u}_{j})\land\omega_{j+1}(\vec{u}_{j},\vec{x}_{j+1})\land  F_{j+1},
\]
where $\omega_{j+1}(\vec{u}_{j},\vec{x}_{j+1})$ is a consistent first-order rewriting of 
\[\
\exists\vec{y}_{j+1}\exists\vec{x}_{j+2}\exists\vec{y}_{j+2}\dotsm\exists\vec{x}_{n}\exists\vec{y}_{n}\lrformula{F_{j+1}\land F_{j+2}\land\dotsm\land F_{n}}.
\]
It follows from~\cite{DBLP:journals/tods/KoutrisW17} that  $\omega_{j+1}(\vec{u}_{j},\vec{x}_{j+1})$ exists, \revision{and can be constructed in linear time in the length~$\card{q}$ of~$q$}.
Then, our desired formula $\frugalformula{q}(\vec{u}_{n})$ is equal to $\psi_{n}(\vec{u}_{n})$:
\[
\frugalformula{q}(\vec{u}_{n})\defeq\psi_{n}(\vec{u}_{n}).
\]
\revision{
It remains to show the quadratic upper bound on the construction of $\psi_{n}(\vec{u}_{n})$.  If $T_{j+1}$ denotes the time for constructing $\psi_{j+1}$, then $T_{0}=1$ and for some constant~$c$, we have  $T_{j+1}\leq T_{j}+c\cdot\card{q}$.
It follows that $T_{n}\leq n\cdot c\cdot\card{q}$.
Since $n\leq\card{q}$, $T_{n}$ is quadratic in the length of~$q$. 
}
\end{proof}

We illustrated the construction with an example.
\begin{example}
Let $q(x,y,z)=\exists x\exists y\exists z\lrformula{R(\underline{x},y)\land S(\underline{y,z},c)}$.
\begin{align*}
    \psi_{1}(x,y) &=\mathsf{true}\land\omega_{1}(x)\land R(\underline{x},y),\\
    \psi_{2}(x,y,z) &=\psi_{1}(x,y)\land\omega_{2}(x,z)\land S(\underline{y,z},c),
\end{align*}
where
\begin{align*}
\omega_{1}(x) & =
\begin{array}[t]{ll}
\exists y R(\underline{x},y)\land\\
\forall y\lrformula{R(\underline{x},y)\rightarrow\exists z\lrformula{S(\underline{y,z},c)\land\forall u\lrformula{S(\underline{y,z},u)\rightarrow u=c}}},
\end{array}\\
\omega_{2}(x,z) & = S(\underline{y,z},c)\land\forall u\lrformula{S(\underline{y,z},u)\rightarrow u=c}.
\end{align*}
Putting all together, with some simplifications:
\[
\psi_{2}(x,y,z)=
\setlength{\arraycolsep}{1pt}
\begin{array}[t]{ll}
& R(\underline{x},y)\land S(\underline{y,z},c)\\
\land & \forall y\lrformula{R(\underline{x},y)\rightarrow\exists z\lrformula{S(\underline{y,z},c)\land\forall u\lrformula{S(\underline{y,z},u)\rightarrow u=c}}}\\
\land & \forall u\lrformula{S(\underline{y,z},u)\rightarrow u=c}.
\end{array}
\]
\myqed
\end{example}

\section{Proof of Lemma~\ref{lem:frugal}}

We introduce some helping constructs and lemmas.

\begin{definition}
    Let $q$ be a query in $\sjfbcq$.
    Let $\db$ be a consistent database instance.
    Let $V$ be a non-empty subset of $\vars{q}$.
    We define $\rifi{q}{\db}{V}$ as the $\subseteq$-minimal set of valuations over~$V$ that contains $\theta$ if $\theta$ can be extended to a valuation $\mu$ over $\vars{q}$ such that $(\db,\mu)\models q$.
\myqed    
\end{definition}

The following helping lemmas extend~\cite[Lemma~4.4]{DBLP:journals/tods/KoutrisW17}.

\begin{lemma}\label{lem:rifi}
    Let $q$ be a query in $\sjfbcq$.
    Let $\db$ be a database instance.
    Let $V$ be a non-empty subset of $\vars{q}$ such that no variable of~$V$ is attacked in~$q$.
    Then, there is a repair $\rep$ of $\db$ such that \[\rifi{q}{\rep}{V} = \bigcap_{\sep\in \repairs{\db}} \rifi{q}{\sep}{V}.\]
\end{lemma}


The following lemma extends Lemma~\ref{lem:rifi} by allowing $V$ to contain variables~$v$ that are attacked in~$q$, provided that $v$ is only attacked by atoms whose corresponding relation in $\db$ is consistent.

\begin{lemma}\label{lem:frugalRepair}
    Let $q$ be a query in $\sjfbcq$.
    Let $\db$ be a database instance.
    Let $V$ be a non-empty subset of $\vars{q}$ such that for every $v\in V$, for every atom $F=R(\underline{\vec{x}}, \vec{y})$ in $q$, if $F\attacks{q}v$, then 
    the $R$-relation of $\db$ is consistent.
    Then, there is a repair $\rep$ of $\db$ such that 
    \[\rifi{q}{\rep}{V} = \bigcap_{\sep \in\repairs{\db}} \rifi{q}{\sep}{V}.\]
\end{lemma}
\begin{proof}
    Let $M$ be a set containing a fresh atom $R'(\underline{\vec{x}}, \vec{y})$ for every atom $R(\underline{\vec{x}}, \vec{y})$ in $q$ such that the $R$-relation of $\db$ is consistent.
    Let $q' = q \cup M$.
    Let $\db'$ be the smallest database instance that includes~$\db$ and includes $\{R'(\underline{\vec{a}}, \vec{c}) \mid R(\underline{\vec{a}}, \vec{c})\in\db\}$ for every atom $R'(\underline{\vec{x}}, \vec{y})$ in $M$.
    It follows from the hypotheses of the lemma that no variable of $V$ is attacked in~$q'$.
    For every repair $\rep$ of $\db$, let~$f(\rep)$ the smallest database instance that includes~$\rep$ and includes $\{R'(\underline{\vec{a}},\vec{b})\mid R(\underline{\vec{a}},\vec{b})\in\db\}$ for every atom $R'(\underline{\vec{x}}, \vec{y})$ in~$M$.
    One can easily verify that every $n$-embedding of~$q$ in $\rep$ is also an $n$-embedding of~$q'$ in~$f(\rep)$, and vice versa.
    Thus, for every repair $\rep$ of $\db$, $\rifi{q}{\rep}{V} = \rifi{q'}{f(\rep)}{V}$.
    Since $\db'\setminus\db$ is consistent by construction, we have that $f:\repairs{\db}\rightarrow\repairs{\db'}$ is a bijective mapping.
    Since no variable of $V$ is attacked in~$q'$, it follows by Lemma~\ref{lem:rifi} that there is a repair $\rep'$ of $\db'$ such that  
       \[\rifi{q'}{\rep'}{V} = \bigcap_{\sep'\in\repairs{\db'}} \rifi{q'}{\sep'}{V}.\]
    From what precedes, it follows
    \begin{align*}
        \rifi{q}{f^{-1}(\rep')}{V} & = \bigcap_{\sep'\in\repairs{\db'}} \rifi{q}{f^{-1}(\sep')}{V}\\
                                   & = \bigcap_{\sep\in\repairs{\db}} \rifi{q}{\sep}{V}
    \end{align*}
    Then, $f^{-1}(\rep')$ is a repair of $\db$ that proves the lemma. 
\end{proof}



The proof of Lemma~\ref{lem:frugal} follows.

\begin{proof}[Proof of Lemma~\ref{lem:frugal}]
The desired result is obvious if there is a repair~$\rep^{*}$ such that $\rep^{*}\models\neg\exists\vec{u}\lrformula{q(\vec{u}})$.
Assume from here on that for every repair $\sep$, there is $\vec{c}$ such that $\sep\models q(\vec{c})$.

For every $i\in\{1,\ldots,n\}$, let $R_{i}$ be the relation name of~$F_{i}$.
Let~$\rep$ be a repair of $\db$.
We show that for every $i\in\{0,1,2,\ldots,n\}$, there is a repair $\sep_{i}$ of~$\db$ such that:
\begin{enumerate}[label=(\alph*)]
    \item\label{it:condSuperFrugal1} every $i$-embedding of~$q$ in $\sep_i$ is an $i$-embedding of~$q$ in $\rep$; and
    \item\label{it:condSuperFrugal2} every $i$-embedding of~$q$ in $\sep_i$ is an $i$-$\forall$embedding of~$q$ in~$\db$.
\end{enumerate}
The construction runs by induction on increasing~$i$.
For the induction basis, $i=0$, let $\sep_{0}$ be an arbitrary repair of~$\db$. 
From our assumption that every repair of $\db$ satisfies $\exists\vec{u}\lrformula{q(\vec{u}})$, it follows that the empty set is both a $0$-embedding in $\rep$ and a $0$-$\forall$embedding in~$\db$, as desired.

For the induction step, $i\rightarrow i+1$, the induction hypothesis is that there is a repair $\sep_{i}$ of~$\db$ that satisfies conditions~\ref{it:condSuperFrugal1} and ~\ref{it:condSuperFrugal2}.
Let $\db^{(i)}$ be the smallest database instance that includes $\sep_{i}$ and includes, for every $j\in\{i+1, i+2,\ldots, n\}$, the $R_{j}$-relation of $\db$. 
For every $m\in\{1,2,\ldots,n\}$, if $F_{m}\attacks{q}v$ with $v\in\bigcup_{j=1}^{i}\vars{F_{j}}$, then $m\leq i$, and hence, by construction, the $R_{m}$-relation of $\db^{(i)}$ is consistent.
Then, by Lemma~\ref{lem:frugalRepair}, there is a repair $\sep^*$ of $\db^{(i)}$ such that 
\begin{equation}\label{eq:rifiintersect}
\rifi{q}{\sep^{*}}{\vec{u}_{i}\vec{x}_{i+1}}=\bigcap_{\tep\in\repairs{\db^{(i)}}}\rifi{q}{\tep}{\vec{u}_{i}\vec{x}_{i+1}}.
\end{equation}

\begin{claim}\label{cla:rifi}
    For every $\nu\in\rifi{q}{\sep^{*}}{\vec{u}_{i}\vec{x}_{i+1}}$,
    we have $(\db,\nu)\cqamodels\{F_{i+1}, F_{i+2}, \ldots, F_{n}\}$.
\end{claim}
\begin{proof}
 Assume for the sake of a contradiction that there is a repair $\tep$ of $\db$ such that 
 \begin{equation}\label{eq:rificontra}
 (\tep,\nu)\not\models\{F_{i+1}, F_{i+2}, \ldots, F_{n}\}.
 \end{equation}
 Let $\tep'$ be the smallest database instance such that
 \begin{itemize}
 \item 
 for every $j\in\{1,2,\ldots,i\}$, $\tep'$ contains all $R_{j}$-facts of $\sep_{i}$; and
 \item 
 for every $j\in\{i+1,i+2,\ldots,n\}$, $\tep'$ contains all $R_{j}$-facts of $\tep$.
 \end{itemize}
Clearly, $\tep'$ is a repair of $\db^{(i)}$.
Thus, by~\eqref{eq:rifiintersect}, $\nu\in\rifi{q}{\tep'}{\vec{u}_{i}\vec{x}_{i+1}}$.
Hence, $(\tep',\nu)\models\{F_{1}, F_{2}, \ldots, F_{n}\}$,
and thus $(\tep',\nu)\models\{F_{i+1}, F_{i+2}, \ldots, F_{n}\}$.
Since $\tep$ and $\tep'$ contain the same $R_{j}$-facts for every $j\in\{i+1,i+2,\ldots,n\}$, it follows $(\tep,\nu)\models\{F_{i+1}, F_{i+2}, \ldots, F_{n}\}$, which contradicts~\eqref{eq:rificontra}.
This concludes the proof of Claim~\ref{cla:rifi}.
\end{proof}

Let $\theta_{1}, \theta_{2}, \ldots, \theta_{g}$ enumerate all $n$-embeddings of~$q$ in~$\sep^{*}$. 
For every $k\in\{1,2,\ldots,g\}$, let $A_{k}$ be the (unique) $R_{i+1}$-fact of $\rep$ that is key-equal to $\theta_{k}(F_{i+1})$.
Let 
\[\sep_{i+1}\defeq\lrformula{\sep^{*}\setminus\set{\theta_{k}(F_{i+1})}_{k=1}^{g}}\cup\{A_{1}, A_{2}, \ldots, A_{g}\}.\] 
Clearly, $\sep_{i+1}$ is a repair of $\db^{(i)}$.

\begin{claim}\label{cla:rifibis}
    $\rifi{q}{\sep^{*}}{\vec{u}_{i}\vec{x}_{i+1}}=\rifi{q}{\sep_{i+1}}{\vec{u}_{i}\vec{x}_{i+1}}$.
\end{claim}
\begin{proof}
    The $\subseteq$-inclusion is straightforward.
    The proof of the $\supseteq$-inclusion is analogous to~\cite[Lemma B.1.]{DBLP:journals/tods/KoutrisW17}, by remarking that each fact in $\{\theta_{k}(F_{i+1})\}_{k=1}^{g}$ is relevant for~$q$ in~$\sep^{*}$. 
\end{proof}

We are now ready to show that conditions~\ref{it:condSuperFrugal1} and~\ref{it:condSuperFrugal2} hold true for~$i+1$. 
To this end, let $\eta$ be an arbitrary $n$-embedding of~$q$ in $\sep_{i+1}$.
Consequently, $\restrict{\eta}{\vec{u}_{i}\vec{x}_{i+1}}\in\rifi{q}{\sep_{i+1}}{\vec{u}_{i}\vec{x}_{i+1}}$.
By Claims~\ref{cla:rifi} and~\ref{cla:rifibis}, it follows
\begin{equation}\label{eq:rifiremainder}
(\db,\restrict{\eta}{\vec{u}_{i}\vec{x}_{i+1}})\cqamodels\{F_{i+1},F_{i+2},\ldots,F_{n}\}.
\end{equation}
Since for every $j\in\{1,2,\ldots,i\}$, the set of $R_{i}$-facts of $\sep_{i+1}$ is identical to that of $\sep_{i}$ (and identical to that of $\db^{(i)}$), it follows that $\restrict{\eta}{\vec{u}_{i}}$ is an $i$-embedding of~$q$ in~$\sep_{i}$.
By the induction hypothesis,
\begin{enumerate}[label=(\Alph*)]
    \item\label{it:rind} $\restrict{\eta}{\vec{u}_{i}}$ an $i$-embedding of~$q$ in $\rep$; and
    \item\label{it:bind} $\restrict{\eta}{\vec{u}_{i}}$ is an $i$-$\forall$embedding of~$q$ in~$\db$. 
\end{enumerate}
From~\ref{it:bind} and~\eqref{eq:rifiremainder}, it follows that $\restrict{\eta}{\vec{u}_{i+1}}$ is an $(i+1)$-$\forall$embedding of~$q$ in~$\db$.

\begin{claim}\label{cla:insider}
There is $k\in\{1,2,\ldots,g\}$ such that $\eta(F_{i+1})=A_{k}$.
\end{claim}
\begin{proof}
    Assume for the sake of a contradiction that  $\eta(F_{i+1})\notin\{A_{1},A_{2},\ldots, A_{g}\}$.
    Then, by the construction of $\sep_{i+1}$ from $\sep^{*}$,
    it follows that~$\eta$ is an $n$-embedding of~$q$ in~$s^{*}$.
    Hence, we can assume $k\in\{1,2,\ldots,g\}$ such that $\eta=\theta_{k}$, and thus $\eta(F_{i+1})=\theta_{k}(F_{i+1})$.
    From $\eta(q)\subseteq\sep_{i+1}$, it follows $\eta(F_{i+1})\in\sep_{i+1}$.
    Hence, $\theta_{k}(F_{i+1})\in\sep_{i+1}$, which can happen only if $\theta_{k}(F_{i+1})=A_{k}$, contradicting $\eta(F_{i+1})\notin\{A_{1},A_{2},\ldots, A_{g}\}$. 
\end{proof}

From~\ref{it:rind} and Claim~\ref{cla:insider}, it follows that 
$\restrict{\eta}{\vec{u}_{i+1}}$ is an $(i+1)$-embedding of~$q$ in $\rep$.
This concludes the induction step.
The proof of~Lemma~\ref{lem:frugal} is concluded by letting $\rep^{*}=\sep_{n}$.
\end{proof}

\section{Equivalence of Superfrugal and $n$-Minimal Repairs}\label{sec:minimality}

\revision{
Let $\exists\vec{u}\lrformula{q(\vec{u})}$ be a Boolean conjunctive query where $q(\vec{u})$ is quantifier-free.
Let $(F_{1}, F_{2}, \ldots, F_{n})$ be a sequence containing all (and only) atoms of~$q(\vec{u})$.
Let $i\in\{1,2,\ldots,n\}$, and $X\defeq\bigcup_{j=1}^{i}\vars{F_{i}}$.
Let~$\db$ be a database instance.
A repair $\rep$ of $\db$ is \emph{$i$-minimal}~\cite{DBLP:conf/icdt/FigueiraPSS23} if there is no repair $\sep$ of $\db$ such that 
\begin{enumerate}
    \item for every valuation $\theta$ over $X$, if $(\sep,\theta)\models q(\vec{u})$, then $(\rep,\theta)\models q(\vec{u})$; 
    and
    \item for some valuation $\mu$ over $X$, $(\rep,\mu)\models q(\vec{u})$ and $(\sep,\mu)\not\models q(\vec{u})$.
\end{enumerate}
An $i$-minimal repair is called $\frugal{X}{q}$-frugal in~\cite{DBLP:journals/tods/KoutrisW17}.

\begin{lemma}\label{lem:minimal}
    Let $q\defeq\exists\vec{u}\lrformula{q(\vec{u})}$ be a query in $\sjfbcq$ with an acyclic attack graph, where $q(\vec{u})$ is quantifier-free.
    Let $(F_{1}, F_{2}, \ldots, F_{n})$ be a topological sort of $q$'s attack graph.
    Let $\db$ be a database instance. Then, 
    \begin{enumerate}[label=(\roman*)]
        \item\label{it:superminimal} every superfrugal repair of $\db$ is $n$-minimal; and
        \item\label{it:minimalsuper} every $n$-minimal repair of $\db$ is superfrugal.
    \end{enumerate}
\end{lemma}
\begin{proof}
\framebox{Proof of~\ref{it:superminimal}}
Let~$\rep$ be a superfrugal repair of $\db$.
Assume for the sake of a contradiction that $\rep$ is not $n$-minimal.
Then, there is a repair $\sep$ of $\db$ such that
\begin{enumerate}[label=(\alph*)]
    \item\label{it:superminimalone} for every valuation $\theta$ over $\vec{u}$, if $(\sep,\theta)\models q(\vec{u})$, then $(\rep,\theta)\models q(\vec{u})$; and
    \item for some valuation $\mu$ over $\vec{u}$, $(\rep,\mu)\models q(\vec{u})$ and $(\sep,\mu)\not\models q(\vec{u})$.
\end{enumerate}
Since $(\rep,\mu)\models q(\vec{u})$ and $\rep$ is superfrugal, it follows that $\mu$ is an $n$-$\forall$embedding of~$q$ in~$\db$.
Let $\ell\in\{1,2,\ldots,n\}$ be the largest index such that for $X\defeq\bigcup_{j=1}^{\ell-1}\vars{F_{j}}$, we have $(\sep,\restrict{\mu}{X})\models q(\vec{u})$.
Let $Y\defeq X\cup\key{F_{\ell}}$.
By the definition of $\forall$embedding, $(\sep,\restrict{\mu}{Y})\models q(\vec{u})$.
Therefore, $\restrict{\mu}{Y}$ can be extended to a valuation $\nu$ over $\vec{u}$ such that $(\sep,\nu)\models q(\vec{u})$.
Since $\mu(F_{\ell})$ and $\nu(F_{\ell})$ are distinct and key-equal, it follows $(\rep,\nu)\not\models q(\vec{u})$, contradicting~\ref{it:superminimalone}.

\framebox{Proof of~\ref{it:minimalsuper}}   
Let $\rep$ be an $n$-minimal repair of $\db$.
By Lemma~\ref{lem:frugal}, there exists a superfrugal repair~$\rep^{*}$ of $\db$ such that for every valuation $\theta$ over~$\vec{u}$, if $(\rep^{*},\theta)\models q(\vec{u})$, then $(\rep,\theta)\models q(\vec{u})$.
Since $\rep$ is $n$-minimal, there is no valuation $\mu$ over~$\vec{u}$ such that $(\rep,\mu)\models q(\vec{u})$ and $(\rep^{*},\mu)\not\models q(\vec{u})$.
It follows that every embedding of $q(\vec{u})$ in $\rep$ is also an embedding of $q(\vec{u})$ in the superfrugal repair $\rep^{*}$, and hence is a $\forall$embedding of~$q$ in~$\db$. 
This concludes the proof.
\end{proof}
}

\section{Proof of Theorem~\ref{the:inexpressible}}

\begin{proof}[Proof of Theorem~\ref{the:inexpressible}]
For a given database instance $\db$, the following are equivalent:
\begin{itemize}
    \item $\interpret{\glbcqa{g()}}{\db}=\bot$;
    \item there is a repair of $\db$ that falsifies  $\exists\vec{y}\lrformula{q(\vec{y})}$.
\end{itemize} 
Assume that the attack graph of $\exists\vec{y}\lrformula{q(\vec{y})}$ has a cycle.
The following problem is known to be $\LOGSPACE$-hard under first-order reductions~\cite{DBLP:journals/tods/KoutrisW17}, and hence not Hanf-local: determine whether a given database instance has a repair that falsifies $\exists\vec{y}\lrformula{q(\vec{y})}$.
It follows from~\cite[Corollary~8.26 and Exercise~8.16]{DBLP:books/sp/Libkin04} that every query in $\withaggr{\FOL}$ is Hanf-local.
It is now correct to conclude that $\glbcqa{g()}$ is not expressible in $\withaggr{\FOL}$. 
\end{proof}

\section{Proof of Lemma~\ref{lem:mcs}}

\begin{proof}[Proof of Lemma~\ref{lem:mcs}]
\framebox{Proof of~\eqref{it:mcs1}.}
Let $\rep$ be a superfrugal repair of $\db$. Let $N$ be the set of all embeddings of $q$ in~$\rep$.  
    Assume for the sake of a contradiction that $N$ is not an MCS of $M$.
    By the definition of superfrugal repair, every element in $N$ is a $\forall$embedding of $q$ in $\db$, and therfore $N \subseteq M$.
    Since $N \models \fdset{q}$ but $N$ is not an MCS of $M$, there is an MCS $N^*$ of $M$ such that $N \subsetneq N^*\subseteq M$.
    Then,  we can assume a valuation $\theta$ in $N^{*}\setminus N$.
    Let $\tuple{F_1, \ldots, F_n}$ be a topological sort of $q$'s attack graph.
    Let $\ell$ be the greatest integer in $\{0, \ldots, n\}$ such that $\theta(\{F_1, \ldots, F_\ell\}) \subseteq \rep$.
    If $\ell = n$, then $\theta \in N$, a contradiction.
    Assume next that $\ell < n$.
    Let $\gamma$ be the restriction of $\theta$ to $\vars{\{F_1, \ldots, F_\ell\}} \cup \key{F_{\ell+1}}$.
    By the definition of $\forall$embedding, we have that $(\rep, \gamma) \models \{F_{\ell+1}, \ldots, F_n\}$.
    Thus, there is an extension $\theta'$ of $\gamma$ such that $\theta' \in N$ and $\{\theta, \theta'\} \not \models \fd{\key{F_{\ell+1}}}{\vars{F_{\ell+1}}}$.
    Since $N \subsetneq N^*$, we also have that $\theta' \in N^*$.
    From $\theta,\theta'\in N^{*}$ and $\{\theta, \theta'\} \not \models \fd{\key{F_{\ell+1}}}{\vars{F_{\ell+1}}}$, it follows $N^* \not \models \fdset{q}$, contradicting that $N^*$ is an MCS.
    We conclude by contradiction that \eqref{it:mcs1} holds true.

    \framebox{Proof of~\eqref{it:mcs2}.}
    Let $N$ be an MCS of $M$.
    Since $N \models \fdset{q}$, there is a repair $\rep_0$ of $\db$ such that every element in~$N$ is an embedding of $q$ in $\rep_0$.
    Conversely, since $N$ is an MCS, it contains every $\forall$embedding $\theta$ of $q$ in $\db$ such that $\theta(q) \subseteq \rep_0$.
    By Lemma~\ref{lem:frugal}, there is a superfrugal repair $\rep$ of~$\db$ such that every embedding of $q$ in $\rep$ is also an embedding of $q$ in $\rep_0$.
    Let $N^*$ be the set of all embeddings of $q$ in $\rep$.
    By the definition of superfrugal repair, every element in $N^*$ is a $\forall$embedding of $q$ in~$\db$.
    Consequently, $N^* \subseteq N$.
    By~\eqref{it:mcs1}, $N^*$ is an MCS of $M$, and therefore $N^{*}$ cannot be a strict subset of~$N$.
    Consequently, $N^* = N$.
    So $\rep$ is a repair of $\db$ such that the set of embeddings of~$q$ in~$\rep$ is exactly~$N$, which concludes the proof.
    Note that the repair that proves~\eqref{it:mcs2} is superfrugal.
\end{proof}

\section{Proof of Theorem~\ref{the:expressible}}\label{sec:proofexpressible}

Let $g()\defeq\AGG(r)\leftarrow q(\vec{u})$ be a numerical query in $\withaggr{\sjfbcq}$ such that~$\agagg$ is monotone and associative, and the attack graph of $\exists\vec{y}\lrformula{q(\vec{u})}$ is acyclic.
We need to show that $\glbcqa{g()}$ is expressible in $\withaggr{\FOL}$.
To improve readability, we will write $\agplus$ to denote  $\agagg$, and use our notational convention $q\defeq\exists \vec{u} \lrformula{q(\vec{u})}$.

Let $\db$ be a database instance.
Let $M$ be the set of all $\forall$embeddings of~$q$ in~$\db$.
By Lemma~\ref{lem:allfo}, $M$ can be calculated in $\FOL$.
By Corollary~\ref{cor:mcs}, 
\begin{equation}\label{eq:cormcsproofs}
    \interpret{\glbcqa{g()}}{\db}=\min_{N\in\mymcs{M}{q}}\agplus\lrformula{\bag{\theta(r)\mid \theta\in N}}.
\end{equation}
The proof of Theorem~\ref{the:expressible} proceeds by showing that the right-hand expression of~\eqref{eq:cormcsproofs} can be expressed in~$\withaggr{\FOL}$.
The technical treatment is subdivided into four sections: after some preliminaries in Section~\ref{subsec:prel}, we show two important helping lemmas, called \emph{Decomposition Lemma} in Section~\ref{subsec:deco}, and \emph{Consistent Extension Lemma} in Section~\ref{subsec:exte}. Finally, the proof of Theorem~\ref{the:expressible} is provided in Section~\ref{subsec:main}.

\subsection{Preliminaries}\label{subsec:prel}

\begin{definition}[Branch]\label{def:branch}
    Let $q$ be a query in $\sjfbcq$ with an acyclic attack graph.
    Let $\tuple{F_1, \ldots, F_n}$ be a topological sort of $q$'s attack graph.
    Let $\ell\in\{0,1,\ldots,n-1\}$.
    Let $\db$ be a database instance.
    Let~$\theta$ be an $\ell$-$\forall$embedding.
    A valuation~$\gamma$ that extends~$\theta$ is called a \emph{branch (of $\theta$)} if both the following hold: 
    \begin{enumerate}[label=(\Alph*)]
        \item\label{it:branchunattacked}  each variable in $\domain{\gamma}\setminus\domain{\theta}$ is unattacked  in the query $\theta(\{F_{\ell+1}, \ldots, F_n\})$; and
        \item\label{it:branchextend}
        $\gamma$ is included in some $n$-$\forall$embedding of $q$ in $\db$.
    \end{enumerate}
    Such a branch is called an \emph{$(\ell + 1)$-$\forall$key-embedding} if $\domain{\gamma}=\domain{\theta}\cup\key{F_{\ell+1}}$.
\myqed
\end{definition}

\begin{lemma}\label{lem:sameKeyFDSameValue}
    Let $q$ be a query in $\sjfbcq$ such that the attack graph of $q$ is acyclic. Let $(F_{1},\ldots,F_{n})$ be a topological sort of $q$'s attack graph.
    Let $\db$ be a database instance.
    Let $\theta_{1}$ and $\theta_{2}$ be two $\ell$-$\forall$embeddings of $q$ in $\db$ such that $\{\theta_{1},\theta_{2}\}\models\fdset{\{F_{1},\ldots,F_{\ell}\}}$.
    For $i\in\{1,2\}$, let $\gamma_{i}$ be a branch of~$\theta_{i}$ such that $\domain{\gamma_{1}}=\domain{\gamma_{2}}$.
    Then, for every $X,Y\subseteq\domain{\gamma_{1}}$, if $\fdset{q}\models\fd{X}{Y}$, then $\{\gamma_{1},\gamma_{2}\}\models\fd{X}{Y}$.
\end{lemma}
\begin{proof}
From~\cite[Lemma~4.4]{DBLP:journals/tods/KoutrisW17}, it follows that for $i\in\{1,2\}$, $(\db,\gamma_{i})\cqamodels\{F_{\ell+1},\ldots,F_{n}\}$.
    Since $\{\theta_1, \theta_2\}\models\fdset{\{F_{1}, \ldots, F_{\ell}\}}$, we can assume a repair $\rep$ of $\db$ such that for $i\in\{1,2\}$,
    \begin{equation}\label{eq:consistentrepair}
    \theta_{i}(\{F_{1},\ldots, F_{\ell}\})\subseteq\rep.
    \end{equation}
    Let $i\in\{1,2\}$.
    Since $(\db,\gamma_{i})\cqamodels\{F_{\ell+1},\ldots,F_{n}\}$, we have $(\rep,\gamma_{i})\models\{F_{\ell+1},\ldots,F_{n}\}$.
    Hence, there exists a valuation $\gamma_{i}^{+}$ over $\vars{q}$ that extends $\gamma_{i}$ such that  $\gamma_{i}^{+}(\{F_{\ell+1}, \ldots, F_{n}\})\subseteq\rep$.
From~\eqref{eq:consistentrepair}, it follows $\gamma_{i}^{+}(\{F_{1},\ldots,F_{n}\})\subseteq\rep$.   Consequently, $\{\gamma_{1}^{+}, \gamma_{2}^{+}\}\models\fdset{q}$.
Let $X,Y\subseteq\domain{\gamma_{1}}$ such that $\fdset{q}\models\fd{X}{Y}$.
Since $\{\gamma_{1}^{+}, \gamma_{2}^{+}\}\models\fdset{q}$, it follows $\{\gamma_{1}^{+}, \gamma_{2}^{+}\}\models\fd{X}{Y}$, hence $\{\gamma_{1}, \gamma_{2}\}\models\fd{X}{Y}$.
\end{proof}

\begin{definition}\label{def:mtheta}
    Let $q$ be a query in $\sjfbcq$.
    Let $\agplus$ be an aggregate operator that is monotone and associative.
    Let $\db$ be a database instance.
    Let $\theta$ be a valuation over some subset of $\vars{q}$ that can be extended to a $\forall$embedding of $q$ in~$\db$.
    We write $\extendthree{\theta}{q}{\db}$ (or simply $\extend{\theta}$ if $q$ and~$\db$ are clear from the context) for the set of $\forall$embeddings of $q$ in~$\db$ that extend $\theta$.
    The rational number~$v$ defined by
    \[ v \defeq \min_{N \in \mymcs{\extend{\theta}}{q}} \bigset{ \agplus\lrformula{\bag{\theta'(r) \mid \theta' \in N}}}\]
    is called the \emph{$\agplus$-minimal value for $\theta$ in $\db$};
    and each MCS $N^*$ of $\extend{\theta}$ satisfying $\agplus\lrformula{\bag{\mu(r) \mid \mu \in N^*}}=v$ is called \emph{$\agplus$-minimal}.
     We say that an MCS $N^*$ of $\extend{\theta}$ is $\agcount$-minimal if for every MCS $N$ of~$M$, we have $\card{N^*}\leq\card{N}$.
\myqed    
\end{definition}

\subsection{Proof of \emph{Decomposition Lemma} (Lemma~\ref{lem:unionIsOptimal})}\label{subsec:deco}

Using the notation $\vec{u}_{\ell}$ defined in Section~\ref{sec:superfrugal}, the following lemma states that for every $(\ell+1)$-$\forall$key-embedding~$\gamma$, if an MCS of $\extend{\restrict{\gamma}{\vec{u}_{\ell}}}$ is restricted to those valuations that include $\gamma$, the result is an  MCS of $\extend{\gamma}$.

 \begin{lemma} \label{lem:keybranchismcs}
    Let $q$, $(F_{1},\ldots,F_{n})$, $\ell$, and $\db$ be as in Definition~\ref{def:branch}.
    Let $\theta$ be an $\ell$-$\forall$embedding of~$q$ in~$\db$, and let~$N$ be an MCS of $\extend{\theta}$.
    Let $\gamma$ be an $(\ell$+1)-$\forall$key-embedding of $q$ in $\db$ that extends~$\theta$.
    Let $N_{\gamma}$ be the subset of $N$ containing all (and only) $\forall$embeddings that extend~$\gamma$.
    Then, $N_\gamma$ is an MCS of~$\extend{\gamma}$.
\end{lemma}
\begin{proof}
    Assume, for the sake of a contradiction, that $N_{\gamma}$ is not an MCS of $\extend{\gamma}$.
    Since $N_{\gamma} \models \fdset{q}$, there is an MCS $N^*$ of $\extend{\gamma}$ such that $N_{\gamma} \subsetneq N^*$.
    We can assume $\mu \in N^*\setminus N_{\gamma}$.
    From $\mu\not\in N_{\gamma}$,
    it follows $\mu \not \in N$.
    Let $k$ be the greatest integer in $\{\ell, \ldots, n\}$ such that for some $\mu'\in N$, we have $\mu(\{F_1, \ldots, F_k\}) = \mu'(\{F_1, \ldots, F_k\})$.
    It is easily verified that such $k$ exists.
    If $k  = n$, then $\mu \in N$, a contradiction.
    Assume from here on that $k < n$.
    Let $\mu_k$ be the restriction of $\mu$ to $\vars{\{F_1, \ldots, F_k\}} \cup \key{F_{k+1}}$.
    Using the same reasoning as in the proof of Lemma~\ref{lem:mcs}, there is a repair~$\rep$ of $\db$ such that~$N$ contains all and only embeddings of $q$ in $\rep$ that extend~$\theta$. 
    Thus, we have that $\mu_k(\{F_1, \ldots, F_k\}) \subseteq \rep$, and by the definition of $\forall$embedding, $(\rep, \mu_k) \models \{F_{k+1}, \ldots, F_n\}$.
    Thus, $N$ contains some extension~$\mu'$ of $\mu_k$ such that $\{\mu, \mu'\} \not \models \fd{\key{F_{k+1}}}{\vars{F_{k+1}}}$.
    Since $\mu' \in N$, it follows that $\mu' \in N_{\gamma}$, and thus $\mu' \in N^*$.
    From $\mu,\mu'\in N^{*}$ and $\{\mu, \mu'\} \not \models \fd{\key{F_{k+1}}}{\vars{F_{k+1}}}$, it follows $N^* \not \models \fdset{q}$, contradicting that $N^*$ is an MCS of $\extend{\gamma_{i}}$.
    We conclude by contradiction that~$N_{\gamma}$ is an MCS of $\extend{\gamma}$.
\end{proof}

Let $v$ be the $\agplus$-minimal value for some  $\ell$-$\forall$embedding $\theta$.
Let $\gamma_1, \ldots, \gamma_k$ enumerate all $(\ell$+1)-$\forall$key-embeddings of $q$ in $\db$ such that each $\gamma_{i}$ extends $\theta$. 
For each $i\in\{1,2,\ldots,k\}$, let $v_{i}$ be the $\agplus$-minimal value for $\gamma_{i}$.
Lemma~\ref{lem:unionIsOptimal} establishes that under some consistency hypothesis, we have $v=\agplus\lrformula{\bag{v_{1},v_{2},\ldots,v_{k}}}$.
Informally, this consistency hypothesis expresses that there is a \emph{single} frugal repair~$\rep$ in which each $v_{i}$ is attained, that is, for each $i\in\{1,2,\ldots,k\}$, $v_{i}=\agplus\lrformula{\bag{\mu(r)\mid\textnormal{$\mu$ is an embedding that extends $\gamma_{i}$ such that $\mu(q)\subseteq\rep$}}}$. 
Lemma~\ref{lem:branchesNoConflict} will demonstrate that this consistency hypothesis can always be satisfied.


\begin{lemma}[Decomposition Lemma]\label{lem:unionIsOptimal}
    Let $q$, $(F_{1},\ldots,F_{n})$, $\ell$, and $\db$ be as in Definition~\ref{def:branch}.
    Let~$\theta$ be an $\ell$-$\forall$embedding of~$q$ in $\db$.
    Let $m$ be the $\agplus$-minimal value for $\theta$ in $\db$ as defined by Definition~\ref{def:mtheta}.
    Let $\gamma_1, \ldots, \gamma_k$ enumerate all extensions of~$\theta$ that are $(\ell$+1)-$\forall$key-embeddings of $q$ in $\db$.
    For $i \in \{1, \ldots, k\}$, let $N_i$ be an $\agplus$-minimal MCS of $\extend{\gamma_{i}}$, and define $v_i \defeq \agplus\lrformula{\bag{\mu(r) \mid \mu \in N_i}}$.
    If $\bigcup_{i=1}^k N_i \models \fdset{q}$, then
    $m=\agplus\lrformula{\bag{v_{1}, v_{2}, \ldots, v_{k}}}$.
\end{lemma}
\begin{proof}
Assume that $\bigcup_{i=1}^k N_i \models \fdset{q}$, which implies that $\bigcup_{i=1}^k N_i$ is an MCS of $\extend{\theta}$.
Let~$N$ be an $\agplus$-minimal MCS of $\extend{\theta}$, i.e,
\begin{equation*}
    m=\agplus\lrformula{\bag{\mu(r) \mid \mu \in N}}.
\end{equation*}
For ease of notation, we define 
\begin{equation}
        \widehat{m} \defeq \agplus\lrformula{\bag{\mu(r) \mid \mu \in \bigcup_{i=1}^k N_i}}.
    \end{equation}
    For each $i \in \{1, \ldots, k\}$, 
    let $N^-_i$ be the subset of $N$ containing all (and only) $\forall$embeddings that extend~$\gamma_i$,
    and define $v^-_i \defeq \agplus\lrformula{\bag{\mu(r) \mid \mu \in N^-_i}}$. 
    Note that $\{N^-_1, N^-_2, \ldots, N^-_k\}$ is a partition of~$N$.
    Since $\agplus$ is associative, it follows that
    \begin{equation}\label{eq:unionIsOptimal2}
        m = \agplus\lrformula{\bag{
           v^-_1,  v^-_2, \ldots, v^-_k
        }},
    \end{equation}
    and
    \begin{equation}\label{eq:unionIsOptimal3}
        \widehat{m} = \agplus\lrformula{\bag{
           v_1, v_2, \ldots, v_k
        }}.
    \end{equation}
    By Lemma~\ref{lem:keybranchismcs}, for every $i \in \{1, \ldots, k\}$, $N^-_i$ is an MCS of $\extend{\gamma_{i}}$.
    By the definition of $\agplus$-minimal MCS, we have that for every $i \in \{1, \ldots, k\}$, $v_i \leq v^-_i$.
    Thus, since $\agplus$ is monotone, we have that
    \begin{equation}\label{eq:unionIsOptimal4}
        \agplus\lrformula{\bag{
           v_1, v_2, \ldots, v_k
        }}
        \leq
        \agplus\lrformula{\bag{
           v^-_1, v^-_2, \ldots, v^-_k
        }}.
    \end{equation}
    By \eqref{eq:unionIsOptimal2}, \eqref{eq:unionIsOptimal3} and \eqref{eq:unionIsOptimal4}, it follows that 
    \begin{equation}\label{eq:unionIsOptimal5}
        \widehat{m}
        \leq
        m.
    \end{equation}
    Since  $m$ is the $\agplus$-minimal value for $\theta$, it follows that 
    \begin{equation}\label{eq:unionIsOptimal6}
        \widehat{m}
        =
        m.
    \end{equation}
    By \eqref{eq:unionIsOptimal3} and~\eqref{eq:unionIsOptimal6}, we can conclude that
    \begin{equation*}
    m
    =
        \agplus\lrformula{\bag{
          v_1, v_2, \ldots, v_k 
        }}.
    \end{equation*}
 This concludes the proof.   
\end{proof}

\subsection{Proof of \emph{Consistent Extension Lemma} (Lemma~\ref{lem:branchesNoConflict})}\label{subsec:exte}

We first show the following helping lemma.

\begin{lemma} \label{lem:IndependentComponent}
Let $q$, $(F_{1},\ldots,F_{n})$, $\ell$, and $\db$ be as in Definition~\ref{def:branch}.
    Let $\gamma$ be a branch of some $\ell$-$\forall$embedding of $q$ in $\db$ such that $\key{F_{\ell+1}} \subseteq \domain{\gamma}$,
    and define $q_{\gamma}\defeq\gamma(\{F_\ell, \ldots, F_n\})$.
    Let $q_1 \subseteq \{F_\ell, \ldots, F_n\}$ such that:
    \begin{enumerate}[label=(\alph*)]
        \item $F_\ell \in q_1$; 
        \item\label{it:wcc} the atoms of $\gamma(q_1)$  form a maximal weakly connected component of the attack graph of $q_{\gamma}$; and 
        \item 
        for every $v\in\vars{q_{\gamma}}$, if $\fdset{q}\models \fd{\key{F_\ell}}{v}$, then $v$ is attacked in $q_\gamma$.
    \end{enumerate}
    Let $q_2 \defeq \{F_\ell, \ldots, F_n\} \setminus q_1$. 
    Then, $\vars{\gamma(q_1)} \cap \vars{\gamma(q_2)} = \emptyset$.
\end{lemma}
\begin{proof}
\proofread{For every $k\in\{\ell+1,\ldots,n\}$, we have $F_{k}\nattacks{q}F_{\ell}$. Consequently, by Lemma~\ref{lem:attacksSameAfterValuation}, $\gamma(F_{\ell})$ is unattacked in~$q_{\gamma}$.}
    Assume $v \in \vars{\gamma(q_1)}$. 
    We show $v\notin\vars{\gamma(q_2)}$.
    Since $v \in \vars{\gamma(q_1)}$ and acyclic attack graphs are known to be transitive, the variable $v$ must occur in an atom of $\gamma(q_1)$ that is either unattacked in the attack graph of $\gamma(q_1)$ or attacked by another atom that itself is unattacked. Hence, there is an  atom $F \in q_1$ such that $\gamma(F)$ is unattacked in~$q_\gamma$, and one of the following holds:
    \minirevision{
    \begin{enumerate}[label=(\alph*)]
    \item\label{it:IndCompOne}
     $v\in\vars{\gamma(F)}$, hence $\fdset{q}\models\fd{\key{F}}{v}$; or
    \item\label{it:IndCompThree}
    $v\in\vars{\gamma(G)}$ for some atom $G\in q_{1}$ such that
    $\gamma(F) \attacks{q_\gamma}\gamma(G)$. 
    \end{enumerate}

    \begin{claim}\label{cla:keyfellf}
        $\fdset{q}\models\fd{\key{F_{\ell}}}{\key{F}}$.
    \end{claim}
    \begin{proof}
        The desired result is obvious if $F=F_{\ell}$.
        Assume $F\neq F_{\ell}$ from here on.
 \proofread{Since the atoms $\gamma(F)$ and~$\gamma(F_{\ell})$ are unattacked in~$q_{\gamma}$ and belong to the same weakly connected component of the attack graph, and since acyclic attack graphs are known to be transitive, there is a sequence of atoms $\tuple{G_1,\ldots, G_k}$ ($k\geq 2$) in $q_1$ such that $G_1 = F_\ell$, $G_k = F$, and such that for every $i \in \{1, \ldots, k-1\}$,
       \begin{itemize}
       \item 
       $\gamma(G_i)\nattacks{q_{\gamma}}\gamma(G_{i+1})$, $\gamma(G_{i+1})\nattacks{q_{\gamma}}\gamma(G_{i})$; and
       \item
       there is an atom $H_i\in q_{1}$ such that $\gamma(G_i)\attacks{q_{\gamma}}\gamma(H_i)$ and $\gamma(G_{i+1})\attacks{q_{\gamma}}\gamma(H_i)$.
       \end{itemize}
        By Lemma~\ref{lem:attacksSameAfterValuation}, for every $i \in \{1, \ldots, k-1\}$, $G_i\nattacks{q}G_{i+1}$, $G_{i+1}\nattacks{q}G_{i}$, $G_i\attacks{q}H_i$, and $G_{i+1}\attacks{q}H_i$.
        By repeated application of Lemma~\ref{lem:attackSameAtom} and logical implication of functional dependencies, we obtain $\fdset{q}\models\fd{\key{F_{\ell}}}{\key{F}}$.
        This concludes the proof of Claim~\ref{cla:keyfellf}.
        }
    \end{proof}
    Assume for the sake of a contradiction that $v\in\vars{\gamma(q_{2})}$.
    We show 
    \begin{equation}\label{eq:vimplied}
        \fdset{q}\models\fd{\key{F_\ell}}{v},
    \end{equation}
    which immediately follows from Claim~\ref{cla:keyfellf} if~\ref{it:IndCompOne} holds true.
    Assume next that~\ref{it:IndCompThree} holds true.
    Since the atoms of $\gamma(q_{1})$ form a maximal weakly connected component of the attack graph of $q_{\gamma}$, it follows that~$v$ is not attacked in $q_{\gamma}$.
    By Lemma~\ref{lem:attacksSameAfterValuation}, it follows $F\nattacks{q}v$, which in turn implies $\fdset{q\setminus\{F\}}\models\fd{\key{F}}{v}$. 
    By Claim~\ref{cla:keyfellf}, we obtain~\eqref{eq:vimplied}. Then, by the hypothesis of the lemma, $v$ is attacked in $q_{\gamma}$, a contradiction.
This concludes the proof. }   
\end{proof}

Let $q=q_{1}\uplus q_{2}$ be as in Lemma~\ref{lem:IndependentComponent}.
The following lemma implies that an $\agplus$-minimal MCS relative to~$q$ can be obtained by taking the cross product of two MCS, $N_{1}$ and $N_{2}$, which are calculated relative to $q_{1}$ and $q_{2}$, respectively. 

\begin{definition}
    Let $\theta$ and $\mu$ be valuations such that $\domain{\theta}\cap\domain{\mu}=\emptyset$.
    We write~$\theta\cdot\mu$ for the valuation over $\domain{\theta}\cup\domain{\mu}$ that extends both $\theta$ and $\mu$.
 \myqed   
\end{definition}

\begin{lemma} \label{lem:WCCIndependent}
    Under the same hypotheses as Lemma~\ref{lem:IndependentComponent}, for every $i\in\{1,2\}$, let $M_{i}$ be the set of $\forall$embeddings of~$\gamma(q_{i})$ in $\db$, and let $N_i$ be an MCS of $M_{i}$ such that 
    \begin{itemize}
        \item $N_i$ is $\agplus$-minimal if $r$ is a variable in $\gamma(q_i)$; and
        \item 
        $N_i$ is $\agcount$-minimal otherwise (i.e., if  $r$ is a variable not in $\gamma(q_i)$ or a constant).
    \end{itemize}
    Then,  
    $
        \{ \gamma \cdot \theta_1 \cdot \theta_2 \mid \theta_1 \in N_1 \myand \theta_2 \in N_2\}
    $
    is an $\agplus$-minimal MCS of $\extend{\gamma}$.
\end{lemma}
\begin{proof}
\proofread{Since $\vars{\gamma(q_1)} \cap \vars{\gamma(q_2)} = \emptyset$ by Lemma~\ref{lem:IndependentComponent}, it follows that for every $\theta_1 \in N_1$ and $\theta_2 \in N_2$, we have that $\gamma \cdot \theta_1 \cdot \theta_2$ is a valid embedding of $q$ in $\db$.}
    Let $N = \{ \gamma \cdot \theta_1 \cdot \theta_2 \mid\theta_1 \in N_1 \myand \theta_2 \in N_2\}$.
    \proofread{
    Clearly, $N$ is a subset of $\extend{\gamma}$.
    Moreover, since $N_1$ and $N_2$ are MCSs of $M_1$ and~$M_2$ respectively, it follows that $N$ is an MCS of $\extend{\gamma}$.
    }
    
    Let $N^*$ be an $\agplus$-minimal MCS of $\extend{\gamma}$, and for $i \in \{1,2\}$, define $N^*_i\defeq\{\restrict{\theta}{\vars{\gamma(q_{i})}}\mid\theta\in N^{*}\}$. 
    Then, $N^* = \{ \gamma \cdot \theta_1 \cdot \theta_2 \mid \theta_1 \in N^*_1 \myand \theta_2 \in N^*_2\}$.    
    \proofread{
        Since $\vars{\gamma(q_1)} \cap \vars{\gamma(q_2)} = \emptyset$ by Lemma~\ref{lem:IndependentComponent}, and since  $N^*$ is an MCS of of $\extend{\gamma}$, it is easily seen that each $N^*_i$ is an MCS of $M_i$  ($i\in\{1,2\}$).
    }
    By the definition of $\agplus$-minimal MCS, it follows that
    \begin{equation}\label{eq:WCCIndependent1}
        \agplus\lrformula{\bag{\mu(r) \mid \mu \in N^*}}
        \leq
        \agplus\lrformula{\bag{\mu(r) \mid \mu \in N}}.
    \end{equation} 
    To show that $N$ is $\agplus$-minimal, we distinguish two cases.
 To ease the notation, for $i\in\{1,2\}$, we define $c_{i}\defeq\card{N_{i}}$ and $c_{i}^{*}\defeq\card{N_{i}^{*}}$.
    \begin{description}
        \item[Case that $r \in \domain{\gamma}$ or $r$ is a constant.]
        Then, $\gamma(r)$ is a constant.
            Indeed, since $\gamma$ is the identity on constants, if $r$ is a constant, then $r=\gamma(r)$.
            We have that for every $i \in \{1,2\}$, $N_i$ is an $\agcount$-minimal MCS of $M_i$.
            With the construct $\copies{i}{t}$ as defined in Definition~\ref{def:dc}, it follows
            \begin{equation}\label{eq:WCCIndependentCase1-1}
                \agplus\lrformula{\bag{\mu(r) \mid \mu \in N}} = \agplus\lrformula{\bag{\copies{\formula{c_{1}*c_{2}}}{\gamma(r)}}},
            \end{equation}        
            and
            \begin{equation}\label{eq:WCCIndependentCase1-2}
                \agplus\lrformula{\bag{\mu(r) \mid \mu \in N^*}} = \agplus\lrformula{\bag{\copies{\formula{c_{1}^{*}*c_{2}^{*}}}{\gamma(r)}}}.
            \end{equation}      
            By the definition of $\agcount$-minimal MCS, for every $i \in \{1,2\}$, $c_{i}\leq c_{i}^{*}$.
            Since $\agplus$ is monotone, 
            \begin{equation}\label{eq:WCCIndependentCase1-3}
                \agplus\lrformula{\bag{\copies{\formula{c_{1}*c_{2}}}{\gamma(r)}}}
                \leq 
                \agplus\lrformula{\bag{\copies{\formula{c_{1}^{*}*c_{2}^{*}}}{\gamma(r)}}}.
            \end{equation}
            From \eqref{eq:WCCIndependentCase1-1}, \eqref{eq:WCCIndependentCase1-2}, and \eqref{eq:WCCIndependentCase1-3}, it follows
            \begin{equation}\label{eq:WCCIndependentCase1-4}
                \agplus\lrformula{\bag{\mu(r) \mid \mu \in N}}
                \leq 
                \agplus\lrformula{\bag{\mu(r) \mid \mu \in N^*}}.
            \end{equation}
            From \eqref{eq:WCCIndependent1} and \eqref{eq:WCCIndependentCase1-4}, it follows 
            \begin{equation}\label{eq:WCCIndependentCase1-5}
                \agplus\lrformula{\bag{\mu(r) \mid \mu \in N}}
                =
                \agplus\lrformula{\bag{\mu(r) \mid \mu \in N^*}}.
            \end{equation}
            It follows that $N$ is an $\agplus$-minimal MCS of $\extend{\gamma}$.
        \item[Case that $r$ is a variable not in $\domain{\gamma}.$]
            Assume, without loss of generality, that $r \in \vars{\gamma(q_1)}$.
            By Lemma~\ref{lem:IndependentComponent}, we have that $r \not \in \vars{\gamma(q_2)}$.
            Thus, $N_1$ is an $\agplus$-minimal MCS of $M_1$, and~$N_2$ is an $\agcount$-minimal MCS of $M_2$.
            It follows that
            \begin{equation}\label{eq:WCCIndependentCase2-1}
                \agplus\lrformula{\bag{\mu(r) \mid \mu \in N}} = \agplus\lrformula{\bag{\copies{c_{2}}{\mu(r)} \mid \mu \in N_1 }},
            \end{equation}
            and
            \begin{equation}\label{eq:WCCIndependentCase2-2}
                \agplus\lrformula{\bag{\mu(r) \mid \mu \in N^*}} = \agplus\lrformula{\bag{\copies{c_{2}^{*}}{\mu(r)}  \mid \mu \in N^*_1 }}.
            \end{equation} 
            Since $\agplus$ is associative, 
            \begin{equation}\label{eq:WCCIndependentCase2-3}
                \agplus\lrformula{\bag{\mu(r) \mid \mu \in N}} = \agplus\lrformula{\bag{ \copies{c_{2}}{\agplus\lrformula{\bag{\mu(r) \mid \mu \in N_1}}} }},
            \end{equation}
            and
            \begin{equation}\label{eq:WCCIndependentCase2-4}
                \agplus\lrformula{\bag{\mu(r) \mid \mu \in N^*}} = \agplus\lrformula{\bag{ \copies{c_{2}^{*}}{\agplus\lrformula{\bag{\mu(r) \mid \mu \in N^*_1}}} }}.
            \end{equation}            
            By the definition of $\agcount$-minimal MCS, $c_{2} \leq c_{2}^{*}$.
            By definition of an $\agplus$-minimal MCS, $\agplus\lrformula{\bag{\mu(r) \mid \mu \in N_1}} \leq \agplus\lrformula{\bag{\mu(r) \mid \mu \in N^*_1}}$.
            Since $\agplus$ is monotone, we obtain that 
            \begin{equation}\label{eq:WCCIndependentCase2-5}
                \agplus\lrformula{\bag{ \copies{c_{2}}{\agplus\lrformula{\bag{\mu(r) \mid \mu \in N_1}}} }}
                \leq 
                \agplus\lrformula{\bag{ \copies{c_{2}^{*}}{\agplus\lrformula{\bag{\mu(r) \mid \mu \in N^*_1}}} }}.
            \end{equation}
            From \eqref{eq:WCCIndependentCase2-3}, \eqref{eq:WCCIndependentCase2-4} and \eqref{eq:WCCIndependentCase2-5}, it follows
            \begin{equation}\label{eq:WCCIndependentCase2-6}
                \agplus\lrformula{\bag{\mu(r) \mid \mu \in N}}
                \leq 
                \agplus\lrformula{\bag{\mu(r) \mid \mu \in N^*}}.
            \end{equation}
            From \eqref{eq:WCCIndependent1} and \eqref{eq:WCCIndependentCase2-6}, it follows
            \begin{equation}
                \agplus\lrformula{\bag{\mu(r) \mid \mu \in N}}
                = 
                \agplus\lrformula{\bag{\mu(r) \mid \mu \in N^*}}.
            \end{equation}
            It follows that $N$ is an $\agplus$-minimal MCS of $\extend{\gamma}$.
    \end{description}
 The proof is now concluded.   
\end{proof}

\begin{lemma}[Consistent Extension Lemma]\label{lem:branchesNoConflict}
    Let $q$, $(F_{1},\ldots,F_{n})$, and $\db$ be as in Definition~\ref{def:branch}.
    Let $\ell\in\{1,2,\ldots,n\}$.
    Let $\agplus$ be an aggregate operator that is monotone and associative.
    Let $\gamma_1, \ldots, \gamma_k$ be a sequence of $\ell$-$\forall$key-embeddings of $q$ in $\db$ such that $\{\gamma_1, \ldots, \gamma_k\} \models \fdset{\{F_1, \ldots, F_{\ell-1}\}}$.
    For every $i \in \{1,\ldots, k\}$, there is an $\agplus$-minimal MCS $N_i$ of $\extend{\gamma_{i}}$ such that $\bigcup_{i=1}^{k} N_i \models \fdset{q}$.
\end{lemma}
\begin{proof}
    For readability, we show the lemma for~$k = 2$. The proof can easily be generalized for~$k>2$.
    
    In the first part of the proof, we show that for $i \in \{1,2\}$, there is an $\ell$-$\forall$embedding $\theta_i$ extending~$\gamma_i$ such that:
    \begin{itemize}
        \item every $\agplus$-minimal MCS of $\extend{\theta_i}$ is also an $\agplus$-minimal MCS of $\extend{\gamma_{i}}$; and
        \item $\{\theta_1, \theta_2\} \models \fdset{\{F_1, \ldots, F_{\ell}\}}$.
    \end{itemize}
    Note that for $i\in\{1,2\}$, $\domain{\gamma_{i}}=\lrformula{\bigcup_{j=1}^{\ell-1}\vars{F_{j}}}\cup\key{F_{\ell}}$, and $\domain{\theta_{i}}=\domain{\gamma_{i}}\cup\notkey{F_{\ell}}$.
    We distinguish two cases:
    \begin{description}
        \item[Case that $\gamma_{1}$ and $\gamma_{2}$ disagree on some variable of $\key{F_{\ell}}$.]
        Let $i \in \{1,2\}$. 
        Let $N^*_i$ be an $\agplus$-minimal MCS of $\extend{\gamma_{i}}$.
        Let $\theta_i$ be the (unique) $\ell$-$\forall$embedding such that every embedding in $N^*_i$ extends $\theta_i$.
        Clearly, every $\agplus$-minimal MCS of $\extend{\theta_i}$ is also an MCS of $\extend{\gamma_{i}}$.
        Since $\gamma_{1}$ and $\gamma_{2}$ disagree on some variable of $\key{F_{\ell}}$, it follows that $\{\theta_1, \theta_2\} \models \fdset{\{F_1, \ldots, F_{\ell}\}}$.

        \item[Case that $\gamma_{1}$ and $\gamma_{2}$ agree on all variables of $\key{F_{\ell}}$.]
        Let $i \in \{1,2\}$.
        Let $\vec{w}$ be a shortest sequence containing each variable $v\in\vars{\gamma_i(\{F_\ell, \ldots, F_n\})}$ 
        such that $v$ is not attacked in $\gamma_i(\{F_\ell, \ldots, F_n\})$ and $\fdset{q} \models \fd{\key{F_\ell}}{v}$.
        Note that $\vec{w}$ is the same for $i=1$ and $i=2$.
        By Lemma~\ref{lem:sameKeyFDSameValue}, there is a sequence of constants~$\vec{a}$, of length $\card{\vec{w}}$, such that for every $j \in \{1,2\}$, for every $\mu \in \extend{\gamma_{j}}$, $\mu(\vec{w}) = \vec{a}$.
        Let $\gamma^+_i$ be the extension of $\gamma_i$ to $\domain{\gamma_i}\cup\vars{\vec{w}}$ such that $\gamma^+_i(\vec{w}) = \vec{a}$.
        It is clear that
        \begin{itemize}
            \item $\gamma^+_1(\vec{w}) = \gamma^+_2(\vec{w})$; and
            \item since $\extend{\gamma_{i}} = \extend{\gamma^+_i}$, every $\agplus$-minimal MCS of $\extend{\gamma^+_i}$ is also an $\agplus$-minimal MCS of $\extend{\gamma_{i}}$.
        \end{itemize}
        
         \begin{claim} \label{claim:stillFD}
             For every variable $v \in \vars{\gamma^+_i(\{F_\ell, \ldots, F_n\})}$, if $\fdset{q}\models\fd{\key{F_\ell}}{v}$, then $v$ is attacked in $\gamma^+_i(\{F_\ell, \ldots, F_n\})$. 
         \end{claim}
         \begin{proof}
            Straightforward from Lemma~\ref{lem:attacksSameAfterValuation} and the construction of $\vec{w}.$            
         \end{proof}

 
        Let $q_{\alpha} \subseteq \{F_{\ell}, \ldots, F_n\}$ such that $F_\ell \in q_{\alpha}$, and the atoms of $\gamma^+_i(q_{\alpha})$  form a maximal weakly connected component of the attack graph of $\gamma^+_i(\{F_\ell, \ldots, F_n\})$.
        Let $q_{\beta} = \{F_\ell, \ldots, F_n\} \setminus q_{\alpha}$. 
       
        \begin{claim}\label{cla:equalqueries}
           For every $v\in\vars{q_{\alpha}}\cap\domain{\gamma^+_1}$, we have $\gamma^{+}_{1}(v)=\gamma^{+}_{2}(v)$. 
        \end{claim}
        \begin{proof}
         \minirevision{
        Let $v\in\vars{q_{\alpha}}\cap\domain{\gamma^+_1}$.
        We need to show $\gamma^+_1(v) = \gamma^+_2(v)$.
        This is obvious if $v\in\vars{\vec{w}}$. 
        Assume for the sake of a contradiction that $v\not\in\vars{\vec{w}}$.
            Then, $v \in \domain{\gamma_1}$. 
            It follows $v \in \vars{\{F_1, \ldots, F_{\ell-1}\}}$, and therefore 
            \begin{equation}\label{eq:nobackwardattack}
            \textnormal{for every\ } k \in \{\ell, \ldots, n\}, F_k \nattacks{q} v.
            \end{equation}
            Define $q_{\gamma^+}\defeq\gamma^+_1(\{F_\ell, \ldots, F_n\})$. It is worth noting that the choice of using $\gamma^+_{1}$ instead of $\gamma^+_{2}$ in the definition of $q_{\gamma^+}$ is unimportant, as the resulting attacks will be the same regardless of the choice.
            Since $v \in \vars{q_{\alpha}}$ and since acyclic attack graphs are transitive, there is an  atom $F \in q_{\alpha}$ such that $\gamma^+_1(F)$ is unattacked in~$q_{\gamma^+}$, and one of the following holds:
            \begin{itemize}
                \item
                $v\in\vars{\gamma^+_1(F)}$, hence $v\in\vars{F}$; or
                \item
                there is $G\in q_{\alpha}$ such that 
                $v\in\vars{\gamma^+_1(G)}$ and
                $\gamma^+_1(F) \attacks{q_{\gamma^+}}\gamma^+_1(G)$, hence $v\in\vars{G}$ and, by  Lemma~\ref{lem:attacksSameAfterValuation}, $F\attacks{q}G$. 
            \end{itemize}     
            Since $F\nattacks{q}v$ by~\eqref{eq:nobackwardattack}, it is correct to conclude $\fdset{q\setminus\{F\}} \models\fd{\key{F}}{v}$. 
            \proofread{By the same reasoning as in the proof of Claim~\ref{cla:keyfellf}, we obtain $\fdset{q\}}\models\fd{\key{F_\ell}}{\key{F}}$.
            Consequently,
                $\fdset{q} \models \fd{\key{F_\ell}}{v}$.
            From~\eqref{eq:nobackwardattack} and Lemma~\ref{lem:attacksSameAfterValuation}, it follows that $v$ is not attacked in $\gamma_i(\{F_\ell, \ldots, F_n\})$.
            Then, by our definition of~$\vec{w}$, we have that $\vec{w}$ contains~$v$, a contradiction.}
            This concludes the proof of Claim~\ref{cla:equalqueries}.
            }
            \end{proof}
         
         Claim~\ref{cla:equalqueries} implies $\gamma^+_1(q_{\alpha}) = \gamma^+_2(q_{\alpha})$.
         Let $N^{(\alpha)}$ be an MCS of~$\extendthree{\gamma_1^+}{\gamma^+_1(q_{\alpha})}{\db}$ such that $N^{(\alpha)}$ is $\agplus$-minimal if $r \in \vars{\gamma^+_1(q_{\alpha})}$, and $\agcount$-minimal otherwise. 
        For $j\in\{1,2\}$, let $N^{(\beta)}_j$ be an MCS of $\extendthree{\gamma_j^+}{\gamma^+_j(q_{\beta})}{\db}$ such that $N^{(\beta)}_j$ is $\agplus$-minimal if $r \in \vars{\gamma^+_j(q_{\beta})}$, and $\agcount$-minimal otherwise.
        By Claim~\ref{claim:stillFD} and Lemma~\ref{lem:WCCIndependent}, the set $N^*_i \defeq \{ \gamma^+_i \cdot \delta \cdot \epsilon \mid \delta \in N^{(\alpha)} \myand \epsilon \in N^{(\beta)}_i\}$ is an $\agplus$-minimal MCS of $\extend{\gamma^+_i}$.
        Let $\theta_i$ be the (unique) $\ell$-$\forall$embedding such that every valuation in~$N^*_i$ extends~$\theta_i$.
        Clearly, every $\agplus$-minimal MCS of $\extend{\theta_i}$ is also an $\agplus$-minimal MCS of $\extend{\gamma^+_i}$, and therefore also an MCS of $\extend{\gamma_{i}}$.
        Since $N^{(\alpha)}$ is the same for $i=1$ and $i=2$, and since $F_\ell \in q_{\alpha}$, it follows that $\theta_1(F_\ell) = \theta_2(F_\ell)$.
        Consequently, $\{\theta_1, \theta_2\} \models \fdset{\{F_1, \ldots, F_{\ell}\}}$.   
    \end{description}
So it is correct to conclude that $\theta_1$, $\theta_2$ with the desired properties exist, which concludes the first part of the proof.
    
We are now ready to prove that the lemma holds for every choice of $\ell$ in the statement of the lemma.
The proof is by induction on decreasing $\ell$.
It is straightforward to see that the lemma holds true when $\ell=n$.
We next show that the lemma holds true when $\ell=g$, assuming that it holds true when $\ell=g+1$.
Let $i \in \{1,2\}$. 
    Let $\gamma^{+, (i)}_1, \ldots, \gamma^{+, (i)}_{k_i}$ enumerate all extensions of $\theta_i$ that are $(g+1)$-$\forall$key-embeddings of $q$ in $\db$, where~$\theta_{i}$ is the $g$-$\forall$embedding whose existence was proved in the first part of the proof. 
    By the induction hypothesis, for every $j \in \{1, \ldots, k_i\}$, there is an $\agplus$-minimal MCS $N^{+, (i)}_j$ of $\extend{\gamma^{+, (i)}_j}$ such that 
    \begin{equation*}\label{eq:heavy}
    \overbrace{\lrformula{\bigcup_{j = 1}^{k_1} N^{+, (1)}_j}}^{N_{1}}\cup\overbrace{\lrformula{\bigcup_{j = 1}^{k_2} N^{+, (2)}_j}}^{N_{2}}\models \fdset{q},
    \end{equation*}
    in which we define $N_{1}$ and $N_{2}$ as shown above.
    It remains to show that $N_{i}$ is an $\agplus$-minimal MCS of $\extend{\gamma_{i}}$.
    Clearly, it suffices to to show that $N_{i}$ is an $\agplus$-minimal MCS of $\extend{\theta_{i}}$.
    To this end, let $m_{i}$ denote the $\agplus$-minimal value for $\theta_{i}$ in $\db$, as defined in Definition~\ref{def:mtheta}.
    Since $N_{i}\models\fdset{q}$, it follows from Lemma~\ref{lem:unionIsOptimal} that 
    \begin{equation*}
        m_{i}
        =
        \agplus\lrformula{\bag{
          v_{1}, v_{2}, \ldots, v_{k_i}
        }},
    \end{equation*}
    where for every $j \in \{1, \ldots, k_i\}$, $v_j \defeq \agplus\lrformula{\bag{\mu(r) \mid \mu \in N^{+, (i)}_j}}$.
    From this, it is correct to conclude that~$N_{i}$ is an $\agplus$-minimal MCS of $\extend{\gamma_{i}}$, which concludes the proof of Lemma~\ref{lem:branchesNoConflict}.
\end{proof}

\subsection{Proof of Theorem~\ref{the:expressible}}\label{subsec:main}

With all these helping lemmas in place, we can now proceed with the proof of Theorem~\ref{the:expressible}.

\begin{proof}[Proof of Theorem~\ref{the:expressible}]
Let $\tuple{F_1, \ldots, F_n}$ be a topological sort of $q$'s attack graph.
By Lemma~\ref{lem:allfo}, there is a $\FOL$ formula that computes the set of $\forall$embeddings.
Let $\theta$ be an $\ell$-$\forall$embedding of $q$ in~$\db$.
We will show, by induction on decreasing $\ell=n, n-1,\ldots,0$, that the $\agplus$-minimal value for $\theta$ in~$\db$ can be computed in $\withaggr{\FOL}$.
Note that for $\ell=0$, we have $\theta=\emptyvaluation$ and, by Corollary~\ref{cor:mcs}, the expression of the $\agplus$-minimal value for $\emptyvaluation$ in $\db$ calculates $\interpret{\glbcqa{g()}}{\db}$.

For the induction basis ($\ell = n$), we have $\extend{\theta}=\{\theta\}=\mymcs{\extend{\theta}}{q}$.
It follows that the $\agplus$-minimal value for $\theta$ in $\db$ is $\agplus\lrformula{\bag{\theta(r)}}$, which can obviously be computed in $\withaggr{\FOL}$.

For the induction step ($\ell+1\rightarrow\ell $), the induction hypothesis is that for every $i\in\{\ell+1,\ell+2,\ldots,n\}$, for every $i$-$\forall$embedding $\theta'$ of $q$ in $\db$, the $\agplus$-minimal value for $\theta'$ in $\db$ can be computed in $\withaggr{\FOL}$.
Let $\gamma_1, \ldots, \gamma_k$ enumerate all $(\ell+1)$-$\forall$key-embeddings of $q$ in $\db$ that extend $\theta$.
By Lemma~\ref{lem:branchesNoConflict}, for each $i\in\{1,2,\ldots,k\}$, we can assume an $\agplus$-minimal MCS $N_{i}$ of $\extend{\gamma_{i}}$ such that $\bigcup_{i=1}^k N_i \models \fdset{q}$.
For every $i \in \{1, \ldots, k\}$, let $v_i \defeq \agplus\lrformula{\bag{\mu(r) \mid \mu \in N_i}}$, that is, $v_i$ is the $\agplus$-minimal value for $\gamma_i$ in~$\db$.
By Lemma~\ref{lem:unionIsOptimal}, $\agplus\lrformula{\bag{v_1, \ldots, v_k}}$ is the $\agplus$-minimal value for $\theta$ in~$\db$.
We show in the next paragraph that  for every $i \in \{1, \ldots, k\}$, $v_i$ can be computed in $\withaggr{\FOL}$.
This suffices to show the theorem, as $\agplus$ can be expressed in $\FOL$.

Let $i \in \{1, \ldots, k\}$.
Let $\theta^+_1, \theta^+_2, \ldots, \theta^+_p$ enumerate all $(\ell+1)$-$\forall$embeddings of $q$ in $\db$ that extend~$\gamma_i$.
For every $j \in \{1, \ldots, p\}$, let $v^+_j$ be the $\agplus$-minimal MCS of $\theta^+_j$.
By the induction hypothesis, each $v^+_j$ can be computed in $\withaggr{\FOL}$.
Since $\fdset{q} \models \fd{\key{F_{\ell+1}}}{\vars{F_{\ell+1}}}$, 
it is clear that $\agmin(\bag{v^+_1, v^+_2, \ldots, v^+_p})$ is the $\agplus$-minimal value for $\gamma_i$ in $\db$.
Since $\agmin$ can be expressed in $\withaggr{\FOL}$, it is correct to conclude that the $\agplus$-minimal value for $\theta$ can be computed in $\withaggr{\FOL}$.
\end{proof}

\section{Proof of Theorem~\ref{the:glbmain}}\label{sec:proofglbmain}

\begin{proof}[Proof of Proof of Theorem~\ref{the:glbmain}]
Let $g()\defeq\agterm{\calF}{\vec{y}}{r}{q(\vec{y})}$ where $q(\vec{y})$ is a self-join-free conjunction of atoms.
If the attack graph of $\exists\vec{y}\lrformula{q(\vec{y})}$ is cyclic, then $\glbcqa{g()}$ is not in $\withaggr{\FOL}$ by Theorem~\ref{the:inexpressible}.
If the attack graph of $\exists\vec{y}\lrformula{q(\vec{y})}$ is acyclic, then $\glbcqa{{g}()}$ is expressible in $\withaggr{\FOL}$ by Theorem~\ref{the:expressible}.

\revision{It remains to establish the upper bounds on the time complexities stated in Theorem~\ref{the:glbmain}. Acyclicity of the attack graph of $\exists\vec{y}\lrformula{q(\vec{y})}$ can be tested using the QuadAttack algorithm described in~\cite[Section~9]{DBLP:journals/tods/Wijsen12}, which runs in quadratic time.
Assuming that the the attack graph of $\exists\vec{y}\lrformula{q(\vec{y})}$ is acyclic, the proof of Theorem~\ref{the:expressible} constructs the expression for $\glbcqa{g()}$ in two steps: first, a $\FOL$ formula for the set of $\forall$embeddings, followed by  a $\withaggr{\FOL}$ formula.
The construction of the former formula is in quadratic time by~Lemma~\ref{lem:allfo}, and the proof of Theorem~\ref{the:expressible} shows that the latter formula can be constructed in linear time (with respect to the length of $g()$).
}
\end{proof}

\section{Proof of Lemma~\ref{lem:dc}}

\begin{proof}[Proof of Lemma~\ref{lem:dc}]
First-order reduction from \problem{2DM} (\problem{2-DIMENSIONAL MATCHING}) which is known to be $\NL$-hard~\cite{DBLP:journals/siamcomp/ChandraSV84}.
\begin{description}
\item[Problem] \problem{2DM}
\item[Instance:] Set $M\subseteq A\times B$, where $A$ and $B$ are disjoint sets having the same number~$n$ of elements. 
\item[Question:] Does $M$ contain a matching, i.e., a subset $M'\subseteq M$ such that $\card{M'}=n$ and no two elements of $M'$ agree in any coordinate?
\end{description}    
Since $\agagg$ has a descending chain, we can assume $s,t\in\rationals_{\geq 0}$ such that $\agagg(\bag{s})>\agagg(\bag{s,t})>\agagg(\bag{s,t,t})>\agagg(\bag{s,t,t,t})>\dotsm$.
Given an instance~$M$ of \problem{2DM}, construct a database instance $\db_{M}$ as follows:
\begin{itemize}
    \item 
    for every $(a,b)\in M$, add $R(\underline{a},b,t)$, $S_{1}(\underline{b},a)$, and  $S_{2}(\underline{b},a)$; and
    \item 
    add $R(\underline{\bot_{A}},\bot_{B},s)$, $S_{1}(\underline{\bot_{B}},\bot_{A})$, and  $S_{2}(\underline{\bot_{B}},\bot_{A})$, where $\bot_{A}$ and $\bot_{B}$ are fresh constants.
\end{itemize}
If $M$ has a matching, then $\db_{M}$ has a repair on which $g()$ returns $\ell\defeq \agagg(\bag{s,\copies{n}{t}})$.
If $M$ has no matching, then every repair of $\db_{M}$ has less than~$n$ embeddings of $g()$'s body, and hence on any repair, $g()$ returns $\agagg(\bag{s,\copies{m}{t}})$ with $m<n$, hence $\agagg(\bag{s,\copies{m}{t}})>\ell$.
Consequently, $\glbcqa{q}$ returns $\ell$ if and only if $M$ has a matching.
\end{proof}

\section{Proof of Lemma~\ref{lem:bdc}}

\begin{proof}[Proof of Lemma~\ref{lem:bdc}]
Proof adapted from a similar proof in~\cite{DBLP:conf/icde/DixitK22}.
First order-reduction from \problem{SIMPLE MAX CUT}, which is known to be $\NP$-hard~\cite{DBLP:conf/stoc/GareyGJ76}.
\begin{description}
\item[Problem] \problem{SIMLE MAX CUT}
\item[Instance:] Graph $G=(V,E)$, positive integer~$K$. 
\item[Question:] Is there a partition of $V$ into disjoint sets $V_{1}$ and~$V_{2}$ such that the number of edges from~$E$ that have one endpoint in $V_{1}$ and one end- 
point in $V_{2}$ is at least~$K$? 
\end{description}    
Since $\agagg$ has a bounded descending chain, we can assume $s,t\in\rationals_{\geq 0}$ such that $\agagg(\bag{s})>\agagg(\bag{s,t})>\agagg(\bag{s,t,t})>\agagg(\bag{s,t,t,t})>\dotsm$.
Let $e=2*\card{E}$.
Moreover, there is $m_{e}\in\rationals_{\geq 0}$ such that for all $j\in\naturals_{>0}$, for all $k',k\in\naturals$ such that $k'\leq k\leq e$,  we have $\agagg(\bag{s,\copies{k'}{t}})<\agagg(\bag{\copies{j}{m_{e}},s,\copies{k}{t}})$. 

Given an instance $G=(V,E)$ of \problem{SIMPLE MAX CUT}, construct a database $\db_{G}$ as follows.
We can assume $E\neq\emptyset$, and that the graph is simple.
Let $d$ be a constant such that $c_{1}\neq d\neq c_{2}$.
\begin{itemize}
    \item for every $v\in V$, $\db_{G}$ contains $S_{1}(\underline{v},c_{1})$ and $S_{1}(\underline{v},d)$;
     \item for every $v\in V$, $\db_{G}$ contains $S_{2}(\underline{v},c_{2})$ and $S_{2}(\underline{v},d)$;
     \item for every edge $\{u,v\}$ in $E$, $\db_{G}$ contains both $T(\underline{u,v},t)$ and $T(\underline{v,u},t)$;
    \item for every $v\in V$, $\db_{G}$ contains $T(\underline{v,v},m_{e})$;
    \item $\db_{G}$ contains $S_{1}(\underline{\bot},c_{1})$, $S_{2}(\underline{\bot},c_{2})$, $T(\underline{\bot,\bot},s)$, where $\bot$ is a fresh constant.  It follows that every repair of $\db_{G}$ satisfies $\exists x\exists y\exists r\lrformula{S_{1}(\underline{x},c_{1})\land S_{2}(\underline{y},c_{2})\land T(\underline{x,y},r)}$.
\end{itemize}
Note that the $T$-relation of $\db_{G}$ is consistent, and hence belongs to every repair.
We show that
$\interpret{\glbcqa{g()}}{\db_{G}}\leq\agagg({\bag{s,\copies{K}{t}}})$ if and only if $G$ is a ``yes''-instance of \problem{SIMPLE MAX CUT}.

\framebox{$\impliedby$}
Assume that $G$ is a ``yes''-instance of \problem{SIMPLE MAX CUT}, as witnessed by a partition of $V$ into $V_{1}$ an $V_{2}$.
Construct a repair $\rep$ as follows:
\begin{itemize}
    \item for every $v\in V_{1}$, $\rep$ contains $S_{1}(\underline{v},c_{1})$ and $R_{2}(\underline{v},d)$;
    \item for every $v\in V_{2}$, $\rep$ contains $R_{2}(\underline{v},c_{2})$ and $S_{1}(\underline{v},d)$.
\end{itemize}
It is easily verified that $\interpret{g()}{\rep}\leq\agagg(\bag{s,\copies{K}{t}})$.

\framebox{$\implies$}
Assume that $\interpret{\glbcqa{g()}}{\db_{G}}\leq\agagg(\bag{s,\copies{K}{t}})$.
We can assume a repair $\rep$ such that $\interpret{g()}{\rep}=\interpret{\glbcqa{g()}}{\db_{G}}$.
Construct $V_{1}$ and $V_{2}$ as follows.
Whenever $\rep$ contains $S_{1}(\underline{v},c_{1})$, then $v\in V_{1}$.
Whenever $\rep$ contains $S_{2}(\underline{v},c_{2})$, then $v\in V_{2}$.
Whenever $\rep$ contains both $S_{1}(\underline{v},d)$ and  $S_{2}(\underline{v},d)$, then $v\in V_{1}$ (the choice is arbitrary). 
Let $j\defeq\card{V_{1}\cap V_{2}}$.
Let $k$ be the number of valuations~$\theta$ over $\{x,y\}$ with $\theta(x)\neq\theta(y)$ such that $(\sep,\theta)\models S_{1}(\underline{x},c_{1})\land S_{2}(\underline{y},c_{2})\land T(\underline{x,y},r)$.
Then, $\interpret{g()}{\rep}=\agagg(\bag{\copies{j}{m_{e}},s,\copies{k}{t}})$.
Clearly, $k\leq e$.
We show $j=0$.
Assume for the sake of a contradiction $j\geq 1$.
Let $\rep'$ be the repair obtained from $\rep$ by replacing $S_{2}(\underline{v},c_{2})$ with $S_{2}(\underline{v},d)$ for every $v\in V_{1}\cap V_{2}$. 
Then, $\interpret{g()}{\rep'}=\agagg(\bag{s,\copies{k'}{t}})$ for some $k'$ with $k'\leq k$.
Then, $\agagg(\bag{s,\copies{k'}{t}})<\agagg(\bag{\copies{j}{m_{e}},s,\copies{k}{t}})$, contradicting that $g()$ reaches a minimum in~$\rep$.
We conclude by contradiction that $j=0$,
hence $\interpret{g()}{\rep}=\agagg(\bag{s,\copies{k}{t}})$.
Since $\agagg(\bag{s,\copies{k}{t}})\leq\agagg(\bag{s,\copies{K}{t}})$,
we have $k\geq K$.
Consequently, the number of edges from~$E$  that have one endpoint in $V_{1}$ and one end-point in $V_{2}$ is at least~$K$. 
\end{proof}

\section{Proof of Theorem~\ref{the:minnew}}

\begin{proof}Proof of Theorem~\ref{the:minnew}]
Assume that the attack graph of $\exists\vec{y}\lrformula{q(\vec{y})}$ is acyclic.
Let $\phi()=\agterm{\agmin}{\vec{y}}{r}{q(\vec{y})}$.
It can be easily verified that for every database instance $\db$, if $\db\not\cqamodels\exists\vec{y}\lrformula{q(\vec{y})}$, then $\interpret{\glbcqa{g()}}{\db}=\bot$; otherwise $\interpret{\glbcqa{g()}}{\db}=\interpret{\phi()}{\db}$. 
These tests can be encoded in~$\withaggr{\FOL}$.
\end{proof}

\section{Proof of Theorem~\ref{the:sepminmax}}

\begin{proof}[Proof of Theorem~\ref{the:sepminmax}]
The ``otherwise'' case,  where the attack graph of $\exists\vec{y}\lrformula{q(\vec{y})}$ has a cycle, follows from Theorem~\ref{the:inexpressible}.
Assume from here on that the attack graph of  $\exists\vec{y}\lrformula{q(\vec{y})}$ is acyclic.
By Theorems~\ref{the:expressible} and~\ref{the:minnew}, we know that $\glbcqa{\MAX(r)\leftarrow q(\vec{y})}$ and $\glbcqa{\MIN(r)\leftarrow q(\vec{y})}$ are expressible in $\withaggr{\FOL}$.

We argue that $\lubcqa{\MIN(r)\leftarrow q(\vec{y})}$ is expressible in $\withaggr{\FOL}$.
Define $(\rationals,\leq')$ such that $r\leq' s$ if and only if $s\leq r$, that is, $\leq'$ reverses the natural order on the rational numbers.
Then, $\lubcqa{\MIN(r)\leftarrow q(\vec{y})}$ relative to~$\leq$ coincides with $\glbcqa{\MAX(r)\leftarrow q(\vec{y})}$ relative to~$\leq'$.  
Let $\varphi()$ be a formula in $\withaggr{\FOL}$ that expresses $\glbcqa{\MAX(r)\leftarrow q(\vec{y})}$ relative to~$\leq$.
An expression for $\lubcqa{\MIN(r)\leftarrow q(\vec{y})}$ can be obtained from $\varphi()$ by reversing the order, in particular, by interchanging $\MAX$ and $\MIN$, as well as $\leq$ and $\geq$. 

By a symmetrical reasoning, it can be argued that $\lubcqa{\MAX(r)\leftarrow q(\vec{y})}$ is expressible in $\withaggr{\FOL}$.
\end{proof}

\section{The Class $\caggforest$}\label{sec:caggforest}

The following definition is borrowed from~\cite[Definition~4.1]{FuxmanThesis}.

\begin{definition}[$\cforest$ and $\caggforest$]
Let $q(\vec{z})$ be a self-join-free conjunctive query, with free variables~$\vec{z}$. The \emph{Fuxman graph} of $q(\vec{z})$ is a directed graph whose vertices are the atoms of~$q(\vec{z})$.
There is a directed edge from an atom $R$ to an atom $S$ if $R\neq S$ and $\notkey{R}$ contains a bound variable that also occurs in~$S$.
The class $\cforest$ contains all (and only)  self-join free conjunctive queries~ $q(\vec{z})$ whose Fuxman graph is a directed forest satisfying, for every directed edge from $R$ to $S$, $\key{S}\setminus\vars{\vec{z}}\subseteq\notkey{R}$.

The class $\caggforest$ contains all numerical queries of one of the following forms:
\begin{itemize}
\item 
$\lrformula{\vec{z},\AGG(u)}\leftarrow q(\vec{z},u)$, where $q(\vec{z},u)$ is in $\cforest$ and $\AGG\in\{\MIN, \MAX, \SUM\}$; or
\item
$\lrformula{\vec{z},\COUNT(*)}\leftarrow q(\vec{z})$, where $q(\vec{z})$ is in $\cforest$.\myqed
\end{itemize}
\end{definition}

\end{document}